\documentclass[a4paper,12pt]{article}
\usepackage{amsthm}{\normalsize }
\usepackage{amsmath}
\usepackage{amssymb}
\usepackage{latexsym}
\usepackage{float}
\usepackage{diagbox} 
\usepackage{threeparttable}
\usepackage{makecell}
\usepackage[textwidth=18cm,textheight=20cm]{geometry}
\usepackage{cite}
\usepackage{booktabs}
\usepackage{multirow}
\usepackage{bm}
\usepackage{array}
\usepackage{mathrsfs}
\usepackage{diagbox}
\usepackage{xcolor}
\usepackage{colortbl,booktabs}
\newtheorem{theorem}{Theorem}[section]
\newtheorem{proposition}[theorem]{Proposition}
\newtheorem{lemma}[theorem]{Lemma}

\newcommand{\tabincell}[2]{
 \begin{tabular}{@{}#1@{}}#2\end{tabular}
}

\usepackage{graphicx} 

\title{Pricing basket options with the first three moments of the basket: log-normal models and beyond}
\begin{document}
\author{\normalsize{$\text{Dongdong Hu}^{ab}, \text{Hasanjan Sayit}^a,\text{Frederi Viens}^c$}\\ 
\footnotesize{$^a \text{Xi'an Jiaotong Liverpool University, Suzhou, China}$}\\ 
\footnotesize{$^b \text{Yiwu Industrial \& Commercial College, Yiwu, China}$} \\
\footnotesize{$^c \text{Rice University, Houston, Texas, USA}$
}}
\date{September 20, 2021}

\maketitle

\abstract{Options on baskets (linear combinations) of assets are notoriously challenging to price using even the simplest log-normal continuous-time stochastic models for the individual assets. The paper \cite{Borovkova}
gives a closed form approximation formula for pricing basket options with potentially negative portfolio weights under log-normal models by moment matching. This approximation formula is conceptually simple, methodologically sound, and turns out to be highly accurate. However it involves solving a system of nonlinear equations which usually produces multiple solutions and which is sensitive  to the selection of initial values in the numerical procedures, making the method computationally challenging. In the current paper, we take the moment-matching methodology in \cite{Borovkova} a step further by obtaining a closed form solution for this non-linear system of equations, by identifying a unary cubic equation based solely on the basket's skewness, which parametrizes all model parameters, and we use it to express the approximation formula as an explicit function of the mean, variance, and skewness of the basket. Numerical comparisons with the baskets considered in \cite{Borovkova} show a very high level of agreement, and thus of accuracy relative to the true basket option price.

In the second half of the paper, we apply the
same moment-matching approach to the case of a class of time-changed models, ones where the log-normal assets in the basket are subordinated by a common random time change, and we obtain a closed form approximation formula for basket option prices under these models also. Our formula in the time-changed case also involves a system of non-linear equations that needs to be solved. We simplify this system of equations and eventually reduce the basket option pricing problem under these models into
the problem of finding a root of a real valued function on the real line. While this equation, which depends explicitly on the moment-generating function of the time-mixing distribution, need not have explicit solutions, we show that in examples of distributions popular in the literatures, the equation can be of polynomial type, or of exponential type, whose numerical resolutions pose no difficulties.  Our  numerical tests, based on financially realistic parameters, show that such closed form approximations to basket prices for this class of  time changed models  also perform highly accurately.
 }

\section{Introduction}
A Basket call option is a financial derivative that gives the owner the right, but not the obligation, to receive the difference between a weighted sum of multiple assets with a fixed constant strike price at the maturity of the option. Basket options allow to hedge against multiple different risks at the same time and they are traded on a large scale in diverse financial markets and in various different forms. For example, in 
agricultural markets, crush spread options, which cover price risk for the producers of secondary commodities like soybean meal and soybean oil, relative to the price of the raw commodity, are traded in large volumes. In energy markets, spread options are traded in a large variety, such as crack spread options where the raw material is crude oil, and the products can be any combination of gasoline, kerosene, and other refined fuels. In electricity markets, spark spread options and its variants are popularly used in hedging both short-term and long-term cross-commodity risks. 

Pricing basket options with only two underlying assets is already considered a challenging issue. Baskets with three or more assets pose notorious option pricing problems, particularly those of spread type, where at least one of the assets has a negative weight (the raw material), and others have positive weights (the products), because the industrial actor seeks price-risk protection, as they acquire that raw material to produced and then sell the transformed material.

As stated in \cite{Li_Deng_2009}, there is growing demand for pricing basket options with more than two underlying assets, whether of spread type or otherwise. Numerical methods for these basket options become extremely slow and largely inapplicable. Because of these problems, various closed form approximations for pricing basket options have been proposed in the past, both in two-asset and multi-asset basket option cases. Many authors have understood the numerical challenges inherent in spread-type basket options; these authors always attempt to handle baskets featuring  positive and negative weights. To their credit, the various approximation methodologies are typically based on simple calibration principles from distributional simplifications, as these are widely considered to be more robust than other numerical methods, and the known properties of the basket can guide the type of distributional approximation. Here are some examples, first with two-asset spreads, then with multiple assets.

In the two-asset spread option case, Shimko\cite{Shimko_1994} advocates 
to approximate the distribution of the spread by a normal random variable due to the fact that the spread of the two underlying assets can be negative. This approach leads to a simple closed form formula but it doesn't perform well in numerical tests. Then Kirk\cite{Kirk1995} combines the second asset with the fixed strike price into a single asset and applies the Margrabe formula to derive a closed form approximation to two-asset spread options under log-normal models. Kirk's formula is relatively accurate and is currently very popular in energy markets. Two subsequent papers, Carmona and Durrelman\cite{Carmona2003a} and Li et al. \cite{Li_Deng_2008}, obtain closed-form approximations to spread options by approximating the exercise boundary of the spread by a linear function and by a second-order Taylor expansion respectively. Then  Bjerksund and Stensland\cite{Bjerk_2014} follow the approach of Carmona and Durrelman\cite{Carmona2003a}, and derive  a closed form lower-bound formula  for the two-asset spread options under log-normal models. 

For baskets with more than two asset, fewer papers have been devoted to pricing their options under log-normal models, see Carmona and Durrleman \cite{Carmona2003b}, Li et al. \cite{Li_Deng_2008}, and Borovkova et al. \cite{Borovkova}. We mentioned that the paper Carmona and Durrleman \cite{Carmona2003b} gives a lower bound for the basket price by linear approximating of the exercise boundary. They manage this for multi-asset baskets as well. Their approach yields a highly accurate lower bound, but it involves solving a system of nonlinear equations, which has a significant computational cost. The paper Li et al. \cite{Li_Deng_2009} also studies pricing problems of multi-asset basket options under log-normal models. According to their terminology, a regular basket is one where the first asset has positive weights and all the remaining assets have negative weights, see their expression (1). A general basket option is called a hybrid basket option in their paper. Their approach in  pricing a general basket option is as follows. They first write down a general basket of assets as the difference of two positively weighted baskets. Then they approximate the first positively weighted basket option by its geometric average and by doing so they reduce a general basket into a regular basket option. After this, they apply a second-order Taylor expansion to the exercise boundary of the regular basket option and obtain a closed-form approximation formula. Their formula's accuracy is sensitive to the model parameters and works well for the case of underlying assets with relatively low volatility. 

Another methodology which turns out to be applicable to basket options was introduced in  Turnbull and Wakeman \cite{Turnbull1991}, for pricing average European options under log-normal models. They use a moment-matching method on  log-normal distributions to  approximate the average European price by using the Edgeworth expansion technique which was proposed in the seminal paper Jarrow and Rudd \cite{JARROW_RUDD}. The paper  Borovkova et al. \cite{Borovkova} contends that this moment-matching methodology is a good idea for pricing basket options under log-normal models. Their paper shows that their moment-matching works well for baskets featuring both positive and negative weights. To obtain their closed form approximation, they use positively shifted or negatively shifted log-normal random variables to approximate the basket by matching the first three moments. According to the numerical tests in their paper, their closed form approximation formula gives prices which are very close to the basket prices obtained by Monte-Carlo simulations, which are used as the benchmark since their fidelity to the true option price can be very high (also see our Section 4 for how this is justified numerically). Their formula is simple and easy to implement, except that its calibration based on moments requires solving a system of nonlinear equations to obtain some parameters in the shifted log-normal variables. Extensive numerical tests show that these nonlinear equations usually give multiple solutions including complex number ones and the numerical procedure to solve this system of equation is very sensitive to the choice of initial values.

There are other papers which apply the moment-matching method for pricing basket options. For example, the paper  Leccadito et al. \cite{Leccadito} proposes an approach involving Hermite polynomials to approximate the non-Gaussian models in the basket: the shifted jump diffusion process. Wu et al.\cite{Wu_Diao2019} approximate the basket option prices under log-normal models by the sum of shifted log-normal distribution with a polynomial expansion (SLNPE).

In the current paper, we investigate the moment-matching approach in the paper Borovkova et al. \cite{Borovkova} further, by delving into the paper's algebraic equations, and by extending the method's reach to a broader model class. 

First, we obtain real-valued closed form solution for the system of non-linear equations that needs to be solved for implementing their closed form approximation formula of the basket option prices. Since we match only three moments, the scaled and shited log-normal model presents a risk of overfitting with its four parameters, to which one responds by positing that the overall scale parameter should be restricted to $\pm 1$.  With this reduction, and letting that scale parameter equal the sign of the basket's skewness, we are able to express the non-linear system's solutions in closed form, by identifying a unary cubic equation whose sole real solution $x$ is a parameter in the closed-form approximation formulas for the model parameters, and then in turn for the basket option price. Interestingly, $x$ only depends on the basket skewness. The solution $x$ is explicit via Cardano's formula. The option price is given explicitly using the shifted scaled log-normal model parameters and the cdf of the standard normal (so-called error function). The four model parameters are given using explicit and rather elementary expressions in terms of the basket's mean, variance, and skewness, and of the explicit solution $x$ to the aforementioned unary cubic equation. Before moving into more general models and into numerics, we  calculate the first order derivatives of the approximation formula with respect to basket's mean, variance, and skewness to find out how sensitive this approximation formula with respect to these three characteristics of the basket. These sensitivity expressions rely on the sensitivity of $x$ with respect to the basket moments. Since $x$ only depends on the skewness $\eta$, and explicitly so, we can calculate $\partial x/\partial \eta$ explicitly in terms of the skewness as well. 

We then generalize the moment-matching method in Borovkova et al.\cite{Borovkova} to the case of baskets with underlying assets that are 
obtained by time-changing their Brownian motions. Because of the scaling property of Brownian motion, the resulting models at maturity are mixtures of shifted scaled log-normals, where the mixing parameter is the squared volatility, and one is free to use any mixing density. These models are also interpretable as shifted Black-Scholes models which are subordinated to random time changes, and are also known as normal variance mixture (NVM) models when considered at maturity. These are popular models in mathematical finance, since they allow modelers to take into account an asset's business time rather than to assume that volatility is constant or that the stock's natural time scale is always the same as the trading floor's standard time scale. 

 Due to all these favorable properties, such time changed models were proposed to model asset returns in the past. For example, in Madan and Seneta \cite{madan-sieneta} multivariate symmetric variance gamma (VG) process that is obtained by subordinating a multivariate Brownian motion without a drift by a common gamma process is proposed as a good model for modelling asset returns. Similarly,
Barndorff-Nielsen \cite{barndorff-nielsen} proposed the multivariate case of the normal inverse Gaussian (NIG)  process using a common subordinator with  inverse Gaussian (IG) marginal distribution as a good model. The extension to asymmetric case, the case of subordinating Brownian motion with non-zero drift, is later extensively studied in Cont and
Tankov \cite{cont-tankov} and Luciano and Schoutens \cite{luciano-schoutens}.

For these type of models, assuming that the mixing (time-change) distribution has been chosen, we continue to advocate for a model with four parameters, of which the overall scale parameter is set to equal the basket skewness, leaving only three parameters to be determined by moment matching. We show that the moment-matching methodology still leads to a system of non-linear equations for the three free paramters, but that these equations are not as straightforward as in the non-time-changed case. We show how they depend explicitly on the characteristic function of the mixing distribution. Assuming these can be solved, we show how to implement a corresponding closed form approximation formula for the basket price; indeed, the standard arbitrage-free call pricing can be solved explicitly using the non-time-changed method, where one then simply takes expectations with respect to the mixing variable, because the thresholding in the pricing theorem, leading to four different formulas, is based on the sign of the skewness, and the sign of the strike price minus the log-normal shift parameter, and thus does not depend on the mixing variable. We are able to simplify this system of equations and eventually reduce the basket pricing problem into a problem of finding a real root of a real valued function on the real line, which is then illustrated explicitly in our Section 4 on numerical illustrations for three different mixing distributions. The reader can consult that section for a discussion of what these three choices of mixing distributions mean in terms of volatility modeling.

The rest of the paper is arranged as follows. In section 2, we discuss basket options under log-normal models and improve the results in  Borovkova et al. \cite{Borovkova}
by providing a closed form solution to their system of non-linear equations. In section 3, we study basket options with underlying assets that are obtained by time changing shifted and scaled log-normal models. In section 4, we present the numerical performances of our approximation formulas, including specific time changes which lead to semi-explicit solutions to the moment-matching parameter calibration, which are easily implemented numerically. Section 5 concludes the paper. Section 6 contains the proofs of our results.



\section{Log-normal Models}
We consider a financial market with $n$ stocks with price processes $S_i(t), i=1, 2, \cdots, n,$ over the time horizon $[0, T]$. The risk-free rate is $r$. We assume the followoing risk-neutral dynamics for the stocks:
\begin{equation}\label{dynamics}
dS_i(t)=rS_i(t)dt+\sigma_{i}dW_i(t), 1\le i\le n,    
\end{equation}
with initial values $S_i(0),$ and volatilities $\sigma_i, i=1,2, \cdots, n$. Here $W_i(t), 1\le i\le n,$ are  standard Brownian motions which are not necessarily independent; we denote by $\rho_{ij}$ the correlation between $W_i(1)$ and $W_j(1)$. We want to price a vanilla call option on the basket of assets $\sum_{i=1}^{n}w_{i}S_{i}(T)$, where the $w_i$'s are the basket weights. The price of the call option on this basket, with strike price $K$, is given by the discounted risk-neutral valuation formula
\begin{equation}\label{basket}
 \Pi=e^{-rT}E\left[\left(\sum_{i=1}^{n}w_{i}S_{i}(T)-K\right)^{+}\right],
\end{equation}
where $(\cdot)^+$ is the positive part function, with
\begin{equation}
\begin{split}\label{lognormal models}
 S_{i}(t)=S_{i}(0)e^{(r-\frac{1}{2}\sigma_{i}^{2})t+\sigma_{i}W_{i}(t)},\quad i=1,\cdots,n,
\end{split}
\end{equation}
since those log-normal models are classically the solutions of the dynamics in (\ref{dynamics}). In (\ref{basket}), the weights $\omega_i, i=1, 2, \cdots, n$ can be positive and negative real numbers, and the expectation operator $E$ is under the risk neutral measure where, as mentioned, the Brownian motions $W_i$ and $W_j$ in the lognormal variables \eqref{lognormal models} have correlations $\rho_{ij}$. We denote by $\Sigma$ the correlation matrix of these $\rho_{ij}$'s. 

To lighten the notation a bit, we denote by $N=(N_{1},\cdots,N_{n})$ an $n-$ dimensional normal random vector with the same distribution as $\frac{1}{\sqrt{T}}(W_{1}(T),\cdots,W_{n}(T))$. Then $(N_{1},\cdots,N_{n})$ are standard normal random variables and their correlation matrix is equal to $\Sigma$. With this new notations, the basket price is immediately written as
\begin{equation}\label{new1}
 \Pi=e^{-rT}E\left[\left(\sum_{i=1}^{n}w_{i}S_{i}(0)e^{(r-\frac{1}{2}\sigma_{i}^{2})T+\sigma_{i}\sqrt{T}N_{i}}-K\right)^{+}\right].
\end{equation}
For further convenience, we introduce the  notation 
\begin{equation}\label{btdef}
 B(T)=:\sum_{i=1}^{n}w_{i}S_{i}(0)e^{(r-\frac{1}{2}\sigma_{i}^{2})T+\sigma_{i}\sqrt{T}N_{i}}.
\end{equation}
which expresses the price at time $T$ of the basket itself. Then the basket's option price in (\ref{new1}) is simply 
\begin{equation}\label{bt}
\Pi=e^{-rT}E(B(T)-K)^+.    
\end{equation}

Because of what is classically known about closed-form formulas for call option prices, one may expect log-normal asset models to be the only cases where such formulas are derivable and are robust to some extent. The strategy is to extend this informal intuition slightly, and to derive  an approximate closed form formula for the price $\Pi$ in \eqref{bt} by approximating the random variable $B(T)$, which is a linear combination of several log-normals, with some appropriate simpler function of a single log-normal random variable. When the basket has positive weights $\omega_i, i=1, 2, \cdots, n,$ L\'evy \cite{Levy_1992} proposed to use a single log-normal random variable to approximate $B(T)$ by matching only the first two moments of the basket asset price and the log-normal random variable. The approach in L\'evy \cite{Levy_1992} gives rather good approximation for the basket price when the weights of the basket are non-negative. When it comes to potentially negative weighted basket case, which as we mentioned in the introduction, is common for spread-type baskets, the moment-matching method to approximate the basket by other types of random variables is rather counter-intuitive as the basket's price can assume negative values. However, despite this fact, the paper \cite{Borovkova}
was able to produce highly accurate approximate closed form basket prices through approximating baskets with potentially negative weights by using shifted log-normal models. In this section, we investigate their approach further. Especially, improving on \cite{Borovkova}, we give a real-valued closed-form solution to the set of non-linear equations in \cite{Borovkova}, and we express the approximate closed-form basket option price as an explicit functional form of the mean, variance, and skewness of the basket itself.

Borovkova et al.\cite{Borovkova},
 propose to use four real-valued parameters $m,s,\tau,c$, and the following approximating random variable 
\begin{equation}\label{4par}
\begin{split}
 c(e^{sN+m}+\tau),
\end{split}
\end{equation}
where $N$ is a standard normal random variable, $e^m$ can be interpreted as a scale parameter ($m$ itself is called a scale parameter in the origial paper \cite{Borovkova}), $s$ is interpreted as shape parameter, and $\tau$ is a location parameter. The parameter $c$ is included to adjust the overall sign of the basket price when weights are not all of the same sign; we will see that it is most efficient to define $c$ as a sign parameter that takes value one when the skewness of the basket is positive and
value negative one if the skewness of the basket is negative. To find the four parameters $c, s, m, \tau,$ the paper  Borovkova et al.\cite{Borovkova} matches the first three moments, via  (10), (11), (12) in their paper, after the value of the sign $c=\pm1$ is determined. While moment matching is an elementary approach, it proves to be highly effective in principle, for basket call pricing, as demonstrated in \cite{Borovkova}. The main challenge of the approach in Borovkova et al.\cite{Borovkova} is that the system of  non-linear equations (10), (11), (12) in their paper needs to be solved numerically. As mentioned in the introduction, the solutions of this system of equations are not unique and selecting proper initial values for the numerical solutions of the system is also challenging. In this section, we address these problems, providing a practical fix by identifying a specific real-valued closed form solution for the non-linear system of equations in Borovkova et al.\cite{Borovkova} which works well for basket call pricing.

The paper Borovkova et al.\cite{Borovkova} considers basket option pricing problems in futures markets. In the current note, without loosing any generality, we consider stock markets and adapt the approach in Borovkova et al.\cite{Borovkova} for pricing basket call options based on stocks.

We begin by calculating the first three moments of the basket based on stocks. These calculations are similar to the calculations (6), (7), (8)  in Borovkova et al.\cite{Borovkova}. For the sake of completeness and for our paper to be self-contained, we present these calculations in a Lemma below.

\begin{lemma}\label{pp1}
Suppose the basket of assets has a price given, as in \eqref{btdef}, by
\begin{equation}
\begin{split}
 B(T)=\sum_{i=1}^{n}w_{i}S_{i}(0)e^{(r-\frac{1}{2}\sigma_{i}^{2})T+\sigma_{i}\sqrt{T}N_{i}}
\end{split}
\end{equation}
and recall that we denote by $\rho_{ij}$ the elements of the correlation matrix $\Sigma$ of $N$. Then the first three moments of $B(T)$ are
\begin{equation}
\begin{split}
 E[B(T)]&=e^{rT}\sum_{i=1}^{n}w_{i}S_{i}(0)\\
 E[(B(T))^{2}]&=e^{2rT}\sum_{i=1}^{n}\sum_{j=1}^{n}w_{i}w_{j}S_{i}(0)S_{j}(0)e^{\rho_{ij}\sigma_{i}\sigma_{j}T}\\
 E[(B(T))^{3}]&=e^{3rT}\sum_{i=1}^{n}\sum_{j=1}^{n}\sum_{k=1}^{n}w_{i}w_{j}w_{k}S_{i}(0)S_{j}(0)S_{k}(0)e^{(\rho_{ij}\sigma_{i}\sigma_{j}+\rho_{ik}\sigma_{i}\sigma_{k}+\rho_{jk}\sigma_{j}\sigma_{k})T}
\end{split}
\end{equation}
\end{lemma}
\begin{proof} The calculations are straightforward.
\end{proof}

For the sake of simplicity of notations, we denote by 
\[
\mu_B(T):=EB(T)
\] the mean of the basket, by  
\[
\sigma_B(T):=\sqrt{E\left[(B(T)-E[B(T)])^{2}\right]}.
\]
the standard deviation of the basket and by  \begin{equation}
\begin{split}
 \eta_B(T):=\frac{E\left[(B(T)-E[B(T)])^{3}\right]}{\sigma_{B}(T)^{3}}.
\end{split}
\end{equation}
the skewness of the basket. The first three moments of the approximating random variable  $c(e^{sN+m}+\tau)$ can be calculated easily as follows
\begin{equation}
\begin{split}
 &M_{1}:=E[c(e^{sN+m}+\tau)]=c(e^{\frac{1}{2}s^{2}+m}+\tau),\\
 &M_{2}:=E[(c(e^{sN+m}+\tau))^{2}]=c^{2}(e^{2s^{2}+2m}+2\tau e^{\frac{1}{2}s^{2}+m}+\tau^{2}),\\
 &M_{3}:=E[(c(e^{sN+m}+\tau))^{3}]=c^{3}(e^{\frac{9}{2}s^{2}+3m}+3\tau e^{2s^{2}+2m}+3\tau^{2}e^{\frac{1}{2}s^{2}+m}+\tau^{3}).
\end{split}
\end{equation}

As explained above, the core idea in Borovkova et al.\cite{Borovkova} is  to estimate the four parameters  $c,s,m,\tau,$ by applying the following moment match 
\begin{equation}\label{moment}
\begin{split}
 M_{1}=E[B(T)],\quad M_{2}=E[(B(T))^{2}],\quad M_{3}=E[(B(T))^{3}].
\end{split}
\end{equation}
The relations (\ref{moment}) give a system of nonlinear equations (10), (11), (12) as in Borovkova et al.\cite{Borovkova} and as mentioned earlier, numerical solutions of this system of equations is challenging. In the following Lemma, we give one real-valued root of this non-linear system of equations.

\begin{lemma}\label{pp2} One set of real-valued solution $c, s, m, \tau, $ that satisfy the moment conditions (\ref{moment}) is given by
\begin{equation}\label{133}
\begin{split}
 c=\mathrm{sgn}(\eta_{B}(T)),\quad s=(\ln(x))^{\frac{1}{2}},\quad m=\frac{1}{2}\ln(\frac{\sigma^{2}_{B}(T)}{x(x-1)}),\quad \tau=\mathrm{sgn}(\eta_{B}(T))\mu_B(T)-\frac{\sigma_{B}(T)}{(x-1)^{\frac{1}{2}}},
\end{split}
\end{equation}
where $\mathrm{sgn}(\cdot)$ is the sign function and
\begin{equation}
\begin{split}
 x=\sqrt[3]{1+\frac{1}{2}\eta_{B}^{2}(T)+\eta_{B}(T)\sqrt{1+\frac{1}{4}\eta_{B}^{2}(T)}}+\sqrt[3]{1+\frac{1}{2}\eta_{B}^{2}(T)-\eta_{B}(T)\sqrt{1+\frac{1}{4}\eta_{B}^{2}(T)}}-1.
\end{split}
\end{equation}
\end{lemma}
\begin{proof}
See the Appendix.
\end{proof}

Now the basket price in  (\ref{new1}) can be approximated by
\begin{equation}\label{aa1}
\begin{split}
 \hat{\Pi}:= e^{-rT}E\left[\left(c(e^{sN+m}+\tau)-K\right)^{+}\right],
\end{split}
\end{equation}
with the parameters $c, s, m, \tau $ given as in Lemma \ref{pp2} above. A further evaluation leads  us to
\begin{equation}\label{pihat}
\hat{\Pi}= \left\{
\begin{array}{ll}
e^{-rT}E\left[\left(e^{sN+m}+\tau-K\right)^{+}\right],\quad &\multirow{1}*{$\eta_{B}(T)>0 $},\\
\specialrule{0em}{1ex}{1ex}
e^{-rT}E\left[\left(-e^{sN+m}-\tau-K\right)^{+}\right],\quad &\multirow{1}*{$\eta_{B}(T)<0$}.
\end{array}
\right.
\end{equation}
We would like to write the expression (\ref{pihat}) as explicitly as possible.
For this, we need to  compare the values of $\tau$ (or $-\tau$) and $K$ according to the sign of $\eta_{B}(T)$. For example, when $\eta_{B}(T)>0$, if $K\leq\tau$, we can directly write the $\hat{\Pi}$ as
\begin{equation}\label{3}
\begin{split}
 \hat{\Pi}= e^{-rT}E\left[e^{sN+m}+\tau-K\right],
\end{split}
\end{equation}
as in this case both $e^{sN+m}$  and $\tau-K$ are non-negative and therefore we can get rid of the positive sign operator inside the expectation in (\ref{aa1}). But if $\eta_{B}(T)>0$ and $K>\tau$ we need to evaluate (\ref{aa1}) differently. Similarly, the case $\eta_{B}(T)<0$ also needs to be treated by dividing into cases. We summarize all these in the following Theorem. From now on we denote by $\Phi(\cdot)$ the cumulative distribution function of standard normal distribution.

\begin{theorem}\label{pp3} The basket option price (\ref{basket}) can be approximated by $\hat{\Pi}$ which is given by
\begin{equation}\label{pp3_1}
\hat{\Pi}=\left\{
\begin{array}{ll}
e^{-rT}(e^{m+\frac{1}{2}s^{2}}+\tau-K),\quad &\multirow{1}*{$c=1,K \leq \tau$},\\
\specialrule{0em}{1ex}{1ex}
e^{-rT+m+\frac{1}{2}s^{2}}\Phi(d_{11})-e^{-rT}(K-\tau)\Phi(d_{12}),\quad &\multirow{1}*{$c=1,K > \tau$},\\
\specialrule{0em}{1ex}{1ex}
0,\quad &\multirow{1}*{$c=-1,K > -\tau$},\\
\specialrule{0em}{1ex}{1ex}
-e^{-rT+m+\frac{1}{2}s^{2}}\Phi(d_{21})+e^{-rT}(-K-\tau)\Phi(d_{22}),\quad &\multirow{1}*{$c=-1,K \le -\tau$},
\end{array}\right.
\end{equation}
where $s, m, \tau$ are given by (\ref{133}) and
\begin{equation}
\begin{split}
 d_{11}&=\frac{-\ln(K-\tau)+m+s^{2}}{s},\\
 d_{12}&=\frac{-\ln(K-\tau)+m}{s},\\
 d_{21}&=\frac{\ln(-K-\tau)-m-s^{2}}{s},\\
 d_{22}&=\frac{\ln(-K-\tau)-m}{s}.
\end{split}
\end{equation}
\end{theorem}
\begin{proof}
See the Appendix.
\end{proof}


Next, we study the sensitiveness of the
approximated basket price $\hat{\Pi}$ with respect to the mean $\mu_B(T)$, the standard deviation $\sigma_{B}(T)$, and the skewness $\eta_{B}(T)$ of basket $B(T)$.
\begin{proposition}\label{th1}
The first order derivatives  of $\hat{\Pi}$ with respect to $\mu_B(T), \sigma_B(T),$ and $\eta_B(T)$ are  given by
\begin{enumerate}
    \item [1.] When $c=1,K\leq\tau$, we have
\begin{equation}
\begin{split}
 \frac{\partial\hat{\Pi}}{\partial \mu_B(T)}=e^{-rT},\quad
 \frac{\partial\hat{\Pi}}{\partial\sigma_{B}(T)}=0,\quad
 \frac{\partial\hat{\Pi}}{\partial\eta_{B}(T)}=0.
\end{split}
\end{equation}
\item [2.] When $c=1,K>\tau$, we have
\begin{equation}
\begin{split}
 \frac{\partial\hat{\Pi}}{\partial \mu_B(T)}&=e^{-rT}\Phi(d_{12}),\\
 \frac{\partial\hat{\Pi}}{\partial\sigma_{B}(T)}&=\frac{e^{-rT}}{\sqrt{x-1}}(\Phi(d_{11})-\Phi(d_{12})),\\
 \frac{\partial\hat{\Pi}}{\partial\eta_{B}(T)}&=\frac{e^{-rT}\sigma_{B}(T)}{2(x-1)^{\frac{1}{2}}}\Big(\frac{-\Phi(d_{11})+\Phi(d_{12})}{x-1}+\frac{\phi(d_{11})}{x(\ln(x))^{\frac{1}{2}}}\Big)\frac{\partial x}{\partial\eta_{B}(T)}.
\end{split}
\end{equation}
\item [3.] When $c=-1,K\geq-\tau$, we have
\begin{equation}
\begin{split}
 \frac{\partial\hat{\Pi}}{\partial \mu_B(T)}=\frac{\partial\hat{\Pi}}{\partial\sigma_{B}(T)}=\frac{\partial\hat{\Pi}}{\partial\eta_{B}(T)}=0.
\end{split}
\end{equation}
\item [4.] When $c=-1,K<-\tau$, we have
\begin{equation}\label{dPideta}
\begin{split}
 \frac{\partial\hat{\Pi}}{\partial \mu_B(T)}&=e^{-rT}\Phi(d_{22}),\\
 \frac{\partial\hat{\Pi}}{\partial\sigma_{B}(T)}&=\frac{e^{-rT}}{\sqrt{x-1}}(-\Phi(d_{21})+\Phi(d_{22})),\\
 \frac{\partial\hat{\Pi}}{\partial\eta_{B}(T)}&=\frac{e^{-rT}\sigma_{B}(T)}{2(x-1)^{\frac{1}{2}}}\Big(\frac{\Phi(d_{21})-\Phi(d_{22})}{x-1}+\frac{\phi(d_{21})}{x(\ln(x))^{\frac{1}{2}}}\Big)\frac{\partial x}{\partial\eta_{B}(T)},
\end{split}
\end{equation}
where $\phi(\cdot)$ is the probability density function of standard normal random variable and 
\begin{equation}
\begin{split}
 d_{11}&=\frac{-\ln(K-\tau)+m+s^{2}}{s},\\
 d_{12}&=\frac{-\ln(K-\tau)+m}{s},\\
 d_{21}&=\frac{\ln(-K-\tau)-m-s^{2}}{s},\\
 d_{22}&=\frac{\ln(-K-\tau)-m}{s},
\end{split}
\end{equation}
with $c, s, m, \tau, x$ given by Lemma \ref{pp2}, and 
\begin{equation}\label{dxdeta}
\begin{split}
&\frac{\partial x}{\partial\eta_{B}(T)}\\&=\frac{1}{3}(1+\frac{1}{2}\eta_{B}^{2}(T)+\eta_{B}(T)(1+\frac{1}{4}\eta_{B}^{2}(T))^{\frac{1}{2}})^{-\frac{2}{3}}(\eta_{B}(T)+(1+\frac{1}{2}\eta_{B}^{2}(T))(1+\frac{1}{4}\eta_{B}^{2}(T))^{-\frac{1}{2}})\\
 &+\frac{1}{3}(1+\frac{1}{2}\eta_{B}^{2}(T)-\eta_{B}(T)(1+\frac{1}{4}\eta_{B}^{2}(T))^{\frac{1}{2}})^{-\frac{2}{3}}(\eta_{B}(T)-(1+\frac{1}{2}\eta_{B}^{2}(T))(1+\frac{1}{4}\eta_{B}^{2}(T))^{-\frac{1}{2}}).
\end{split}
\end{equation}

\end{enumerate}
\end{proposition}
\begin{proof}
See the Appendix.
\end{proof}
The sensitivity of the basket option price to the basket's skewness can be investigated directly from equations (\ref{dPideta}) and (\ref{dxdeta}). The limiting behavior for small skewness $\eta=\eta_{B}(T)$ is particularly informative. Using direct calculations, we find that, for small $\eta$, $x=1+(1/9)\eta^2+O(\eta^3)$ while $\partial x/\partial \eta = (2/9)\eta+O(\eta^2)$. Consequently we find, for small $\eta=\eta_{B}(T)$,
\begin{equation}
 \frac{\partial\hat{\Pi}}{\partial\eta_{B}(T)}=3e^{-rT}\sigma_{B}(T)(\Phi(d_{21})-\Phi(d_{22}))\times \frac{1}{\eta^2}
 +e^{-rT}\sigma_{B}(T)\phi(d_{21})\times\frac{1}{\eta}+O(1).
\end{equation}
The constants in front of the terms $1/\eta^2$ and $1/\eta$ are not dependent on $\eta$. This means that the basket option price is very sensitive to the basket's skewness $\eta$ in the limit of small skewness. Thus the transition in pricing from baskets with  symmetric distributions to baskets with non-zero skews occurs rather abruptly. 

This sensitivity shows that the use of a method which calibrates the basket skewness directly, such as a moment matching method as we adopt in this paper, would most likely be relevant in practice, for instance, for basket spreads with more than two assets, such as soy crush spreads, or for baskets whose calls cannot be interpreted as ordinary exchange options, such as oil crack spreads. Indeed those spreads are skewed though not heavily so, falling in a range where sensitivity to the skew is very significant.


\section{Time Changed Models}
The above section presents a compact, closed -form 
approximation formula for basket option prices under log-normal models, improving on the results in \cite{Borovkova}. The salient 
feature of the approach in the above section is to match  baskets moments in the presence of potentially negative and positive weights. These  weights of both signs result in random variable models for the baskets that may take both negative and positive values with positive probability. As mentioned, this moment-matching approach produces quite accurate closed-form approximations for basket option prices under realistic log-normal models. 

The purpose of the current section is to investigate the outcome of the moment-matching approach employed in the above section to other classes of models. In this section we restrict our attention to models obtained by time-changing a Brownian motion without a drift, using various subordinators. By leaving the case of a time-changed model with drift to another study, we get to the core of the issue of time-changed log-normal models, avoiding additional burdens on the precision of the closed form approximation formulas.  Calculations which are similar to what is carried out in this section can be performed for Brownian models with non-zero drift. However, we performed extensive numerical tests in the cases of non-zero drifts, showing that the resulting approximations are not necessarily numerically robust.

We begin with a few definitions. An $n$-dimensional random variable $X$ is said to have a normal variance mixture (NVM) structure if
\begin{equation}\label{1}
X\overset{d}{=}\mu+\sqrt{Y}AN^{n},
\end{equation}
where $\mu\in R^{n}$ is a constant vector, $A$ is an
$n\times n$ matrix,  $N^{n}$ is an $n-$dimensional standard normal random variable, 
and $Y$ is a non-negative scalar random variable independent from $N^n$.

Let $B_t$ be an $n-$dimensional Brownian motion with covariance matrix $\Sigma$ and let $T_t$ be any subordinator (i.e. any a.s. positive increasing process), then, at each time point $t>0$, the variable $\mu+B_{T_t}$ is equal in distribution  to a NVM distribution as in (\ref{1}),  see Barndorff-Nielsen\cite{Barndorff},  Eberlein and Keller\cite{Eberlein}, Bingham and Kiesel\cite{Bingham}, Schoutens\cite{Schoutens}, Prause\cite{Prause}, Raible\cite{Raible}. More specifically, we have
\begin{equation}
 \mu+B_{T_t}\overset{d}{=}\mu+\sqrt{T_t}AN^n,\;   t>0,
\end{equation}
where $A=\Sigma^{\frac{1}{2}}$. 

In this section we consider $n$ risky assets with price processes of the form $S_i(t)=S_i(0)e^{\nu_it+L_i(t)}$, $ i=1, 2, \cdots, n$, where $L(t)=(L_1(t), L_2(t), \cdots, L_n(t))$ is given by $L(t)=\mu+B_{Y_t}$ for some vector $\mu \in R^n$, some subordinator $Y_t$, some $n$ dimensional Brownian motion $B$ with co-variance matrix $\Sigma$. Note that the same $Y$ is used to subordinate all components of $B$. The  parameters $\nu_i, i=1, 2, \cdots, n$ can be computed via the martingale condition (see equation (4) of Jurczenko et al.\cite{Emmanuel}),
\begin{equation}\label{29}
\begin{split}
 E[e^{-rt}S_i(t)]=S_{i}(0),\quad i=1,\cdots,n,
\end{split}
\end{equation}
Denoting the $i$th row vector of $A=\Sigma^{\frac{1}{2}}$ by $A_i$ (where by definition $A^T A=\Sigma$), and the variance of $A_i^TB_1$ by $\sigma_i^2$ for each $i=1,2, \cdots, n$, from (\ref{29}) we get explicit expressions for $\nu_i, i=1, 2, \cdots, n$ as follows
\begin{equation}
\begin{split}
 \nu_{i}&=-\mu_{i}-\frac{1}{t}\ln(E[e^{\sigma_{i}\sqrt{Y_{t}}N}])\\
 &=-\mu_{i}-\frac{1}{t}\ln\left(E\left[E[e^{\sigma_{i}\sqrt{Y_{t}}N}|Y_{t}]\right]\right)\\
 &=-\mu_{i}-\frac{1}{t}\ln\left(E\left[e^{\frac{\sigma_{i}^{2}}{2}Y_{t}}\right]\right)\\
 &=-\mu_{i}-\frac{1}{t}\ln\left(\phi_{Y_{t}}(\frac{\sigma_{i}^{2}}{2})\right),\quad i=1,\cdots,n,
\end{split}
\end{equation}
where $\phi_{Y_{t}}(\cdot)$ is the moment generating function of $Y_{t}$. 

Therefore, the marginal distribution at time $t$ of the $i$th risky asset can be expressed as 
\begin{equation}
\begin{split}
 S_{i}(t)\overset{d}{=}\frac{S_{i}(0)}{\phi_{Y_{t}}(\frac{\sigma_{i}^{2}}{2})}e^{rt+\sigma_{i}Y_{t}N_i}\quad i=1,\cdots,n,
\end{split}
\end{equation}
where $N_i$ is the $i$th component of an $n-$dimensional standard Normal vector $N_n$. Our goal is to obtain an approximate closed-form formula for the basket price $\Pi=e^{-rT}E(B(T)-K)^+$, where
\begin{equation}\label{111}
\begin{split}
 B(T)=\sum_{i=1}^{n}\frac{w_{i}S_{i}(0)}{\phi_{Y_{T}}(\frac{\sigma_{i}^{2}}{2})}e^{rT+\sigma_{i}\sqrt{Y_{T}}N_{i}}.
\end{split}
\end{equation}
For this purpose, we would like to approximate $B(T)$ by some random variable. Following the idea of the approach of Borovkova et al.\cite{Borovkova}, we propose the random variable
\begin{equation}
\begin{split}
 W=c(e^{s\sqrt{Y_{T}}N+m}+\tau)
\end{split}
\end{equation}
to approximate the basket $B(T)$ in (\ref{111}). This model can be labeled as a  log-normal shape-mixture model. First, we specify the possible  value ranges of these parameters: since $N$ is symmetric about the origin, we can assume $s\in R_{+}$, parameters $m,\tau\in R$ can be any real numbers, and $c$ is the sign parameter, which can only be $1$ or $-1$.

Next, we will explain how to use the random variable $W=c(e^{s\sqrt{Y_{T}}N+m}+\tau)$ to approximate the basket $B(T)$ by the moment-matching method. First, we give the formulas of the first three moments of the basket distribution.
\begin{lemma}\label{p1}
Suppose the basket of assets has a price given by
\begin{equation}
\begin{split}
 B(T)=\sum_{i=1}^{n}\frac{w_{i}S_{i}(0)}{\phi_{Y_{T}}(\frac{\sigma_{i}^{2}}{2})}e^{rT+\sigma_{i}\sqrt{Y_{T}}N_{i}}
\end{split}
\end{equation}
where $N$ has correlation matrix $\Sigma$, and $\phi_{Y_{T}}(\cdot)$ is the moment-generating function of $Y_{T}$. Then the first three moments of the basket distribution are
\begin{equation}
\begin{split}
 E[B(T)]&=e^{rT}\sum_{i=1}^{n}w_{i}S_{i}(0)\\
 E[(B(T))^{2}]&=e^{2rT}\sum_{i=1}^{n}\sum_{j=1}^{n}\frac{\phi_{Y_{T}}(\frac{1}{2}\sigma_{i}^{2}+\rho_{ij}\sigma_{i}\sigma_{j}+\frac{1}{2}\sigma_{j}^{2})}
 {\phi_{Y_{T}}(\frac{1}{2}\sigma_{i}^{2})\phi_{Y_{T}}(\frac{1}{2}\sigma_{j}^{2})}w_{i}w_{j}S_{i}(0)S_{j}(0)\\
 E[(B(T))^{3}]&=e^{3rT}\sum_{i=1}^{n}\sum_{j=1}^{n}\sum_{k=1}^{n}\frac{\phi_{Y_{T}}(\frac{1}{2}\sigma_{i}^{2}+\frac{1}{2}\sigma_{j}^{2}+\frac{1}{2}\sigma_{k}^{2}+\rho_{ij}\sigma_{i}\sigma_{j}+\rho_{ik}\sigma_{i}\sigma_{k}+\rho_{jk}\sigma_{j}\sigma_{k})}
 {\phi_{Y_{T}}(\frac{1}{2}\sigma_{i}^{2})\phi_{Y_{T}}(\frac{1}{2}\sigma_{j}^{2})\phi_{Y_{T}}(\frac{1}{2}\sigma_{k}^{2})}\\
 \times& w_{i}w_{j}w_{k}S_{i}(0)S_{j}(0)S_{k}(0).
\end{split}
\end{equation}
\end{lemma}
\begin{proof}
The proof is straightforward.
\end{proof}

As in the case of log-normal models, we calculate the skewness $\eta_{B}(T)$ of the basket distribution and we set $c=1$ if $\eta_{B}(T)>0$, and $c=-1$ if $\eta_{B}(T)<0$. Next, we can calculate the first three moments of the approximating variable $W$.
\begin{lemma}\label{p2}
Suppose the approximating variable is $W=c(e^{s\sqrt{Y_{T}}N+m}+\tau)$, then if $\eta_{B}(T)>0$, the first three moments of $W$ are
\begin{equation}
\begin{split}
 &M_{1}(T):=E[W]=e^{m}\phi_{Y_{T}}(\frac{1}{2}s^{2})+\tau\\
 &M_{2}(T):=E[W^{2}]=e^{2m}\phi_{Y_{T}}(2s^{2})+2\tau e^{m}\phi_{Y_{T}}(\frac{1}{2}s^{2})+\tau^{2}\\
 &M_{3}(T):=E[W^{3}]=e^{3m}\phi_{Y_{T}}(\frac{9}{2}s^{2})+3\tau e^{2m}\phi_{Y_{T}}(2s^{2})+3\tau^{2}e^{m}\phi_{Y_{T}}(\frac{1}{2}s^{2})+\tau^{3},
\end{split}
\end{equation}
while if $\eta_{B}(T)<0$, the first three moments of $W$ are
\begin{equation}
\begin{split}
 &M_{1}(T):=E[W]=-e^{m}\phi_{Y_{T}}(\frac{1}{2}s^{2})-\tau\\
 &M_{2}(T):=E[W^{2}]=e^{2m}\phi_{Y_{T}}(2s^{2})+2\tau e^{m}\phi_{Y_{T}}(\frac{1}{2}s^{2})+\tau^{2}\\
 &M_{3}(T):=E[W^{3}]=-e^{3m}\phi_{Y_{T}}(\frac{9}{2}s^{2})-3\tau e^{2m}\phi_{Y_{T}}(2s^{2})-3\tau^{2}e^{m}\phi_{Y_{T}}(\frac{1}{2}s^{2})-\tau^{3}
\end{split}
\end{equation}
where in both cases, $\phi_{Y_{T}}(\cdot)$ is the moment generating function of $Y_{T}$.
\end{lemma}
\begin{proof}
The calculations are straightforward.
\end{proof}

After getting the first three moments of the basket distribution and of the approximating (log-normal shape-mixture) random variable, as in the pure log-normal case, the strategy is to match these moments as follows
\begin{equation}\label{38}
\begin{split}
 E[B(T)]&=M_{1}(T)\\
 E[(B(T))^{2}]&=M_{2}(T)\\
 E[(B(T))^{3}]&=M_{3}(T)
\end{split}
\end{equation}
and to obtain a solution for the three parameters $m,s,\tau$ that allows this match, and which provides both some numerical stability and a good approximation. 

Our method to compute a closed-form approximation of the basket option price is first to propose a solution to (\ref{38}) for the parameters $c, s, m, \tau$, which we do in the following Lemma, and then to use those parameters in a risk-neutral pricing formula.

\begin{lemma}\label{ppp1}
A set of four parameters $c, s, m, \tau,$ which satisfy (\ref{38}) is given by
\begin{equation}
\begin{split}
 &c=\mathrm{sgn}(\eta_{B}(T)),\quad s=\sqrt{x},\quad m=\frac{1}{2}\ln(\frac{\sigma^{2}_{B}(T)}{\phi_{Y_{T}}(2x)-(\phi_{Y_{T}}(\frac{1}{2}x))^{2}}),\\ &\tau=\mathrm{sgn}(\eta_{B}(T))\mu_B(T)-\frac{\phi_{Y_{T}}(\frac{1}{2}x)}{(\phi_{Y_{T}}(2x)-(\phi_{Y_{T}}(\frac{1}{2}x))^{2})^{\frac{1}{2}}}\sigma_{B}(T),
\end{split}
\end{equation}
where $\mathrm{sgn}(\cdot)$ is the signum function, $\phi_{Y_{T}}(\cdot)$ is the moment-generating function of $Y_{T}$,  and $x$ is the strictly positive real root of
\begin{equation}\label{39}
\begin{split}
 \phi_{Y_{T}}(\frac{9}{2}x)-3\phi_{Y_{T}}(\frac{1}{2}x)\phi_{Y_{T}}(2x)+2(\phi_{Y_{T}}(\frac{1}{2}x))^{3}-|\eta_{B}(T)|(\phi_{Y_{T}}(2x)-(\phi_{Y_{T}}(\frac{1}{2}x))^{2})^{\frac{3}{2}}=0,
\end{split}
\end{equation}
that satisfies  $\phi_{Y_{T}}(2x)-(\phi_{Y_{T}}(\frac{1}{2}x))^{2}\neq 0$.
\end{lemma}
\begin{proof}
See the Appendix.
\end{proof}

With these parameters in hand, we can calculate the option price by using the approximating (log-normal shape-mixture) random variable.

As in the pure log-normal case, we write the approximate basket price in cases depending on the sign of the skewness, because its sign determines the value of $c=\pm1$, as follows

\begin{equation}
\bar{\Pi}=: \left\{
\begin{array}{ll}
e^{-rT}E\left[\left(e^{s\sqrt{Y_{T}}N+m}+\tau-K\right)^{+}\right],\quad &\multirow{1}*{$\eta_{B}(T)>0$},\\
\specialrule{0em}{1ex}{1ex}
e^{-rT}E\left[\left(-e^{s\sqrt{Y_{T}}N+m}-\tau-K\right)^{+}\right],\quad &\multirow{1}*{$\eta_{B}(T)<0$}.
\end{array}\right.
\end{equation}

Then following a similar analysis as in the pure log-normal case, we obtain the approximate closed-form formula for the basket call option price. The following theorem results.
\begin{theorem}\label{p3}
The basket call option price can be approximated as follows
\begin{enumerate}
\item [a)] When the basket's skewness $\eta_{B}(T)>0$, then the basket's call option price is approximated as
\begin{equation}\label{p3_1}
\begin{split}
 \bar{\Pi}= e^{-rT}E\left[\left(e^{s\sqrt{Y_{T}}N+m}+\tau-K\right)^{+}\right].
\end{split}
\end{equation}
If $K\leq\tau$, then that approximate pricing formula becomes
\begin{equation}\label{p3_2}
\begin{split}
 \bar{\Pi}= e^{-rT}(\phi_{Y_{T}}(\frac{1}{2}s^{2})+\tau-K),
\end{split}
\end{equation}
where $\phi_{Y_{T}}(\cdot)$ is the moment generating function of $Y_{T}$. If $K>\tau$, then that approximate pricing formula becomes
\begin{equation}\label{p3_3}
\begin{split}
 \bar{\Pi}= e^{-rT}\left(E\left[e^{\frac{1}{2}s^{2}Y_{T}+m}\Phi(d_{11}(Y_{T}))\right]-(K-\tau)E\left[\Phi(d_{12}(Y_{T}))\right]\right),
\end{split}
\end{equation}
where the expectation is respect to the distribution of $Y_{T}$ and
\begin{equation}\label{p3_4}
\begin{split}
 d_{11}(Y_{T})&=\frac{-\ln(K-\tau)+m+s^{2}Y_{T}}{s\sqrt{Y_{T}}},\\
 d_{12}(Y_{T})&=\frac{-\ln(K-\tau)+m}{s\sqrt{Y_{T}}}.
\end{split}
\end{equation}

\item [b)]  When the basket's skewness $\eta_{B}(T)<0$, then the basket's call option price is approximated as
\begin{equation}\label{p4_1}
\begin{split}
 \bar{\Pi}= e^{-rT}E\left[\left(-e^{s\sqrt{Y_{T}}N+m}-\tau-K\right)^{+}\right].
\end{split}
\end{equation}
If $K\geq-\tau$, then that approximate pricing formula becomes
\begin{equation}\label{p4_2}
\begin{split}
 \bar{\Pi}= 0,
\end{split}
\end{equation}
If $K<-\tau$, then that approximate pricing formula becomes
\begin{equation}\label{p4_3}
\begin{split}
 \bar{\Pi}= e^{-rT}\left(-E\left[e^{\frac{1}{2}s^{2}Y_{T}+m}\Phi(d_{21}(Y_{T}))\right]+(-K-\tau)E\left[\Phi(d_{22}(Y_{T}))\right]\right),
\end{split}
\end{equation}
where the expectation is respect to the distribution of $Y_{T}$ and
\begin{equation}\label{p4_4}
\begin{split}
 d_{21}(Y_{T})&=\frac{\ln(-K-\tau)-m-s^{2}Y_{T}}{s\sqrt{Y_{T}}},\\
 d_{22}(Y_{T})&=\frac{\ln(-K-\tau)-m}{s\sqrt{Y_{T}}}.
\end{split}
\end{equation}

\end{enumerate}
\end{theorem}
\begin{proof}
See the Appendix.
\end{proof}

In the second part of Section 4, we investigate specific examples of mixing distributions for $Y_T$, including an exponential, a higher-order Gamma, and an inverse-Gamma distribution. These lead to semi-explicit expressions for the pricing formulas in the above theorem, whose algebraic equation \eqref{39} for the key parameter $x$ we present explicitly, from which all formulas are readily evaluated numerically.

\section{Numerical Results}
In this section, we examine the numerical performances of our formulas in the previous sections. 

First, we compare the performances of using the closed form  solutions of the parameters in Lemma \ref{pp2} directly in the approximate basket option pricing formula 
in Theorem \ref{pp3} versus calculating these parameters
from the non-linear system of equations in \cite{Borovkova} and implementing them on the approximate basket option pricing formula. The model parameters in  Table 1 below are exactly the same as the model parameters used in Table 3 of Borovkova et al.\cite{Borovkova}. Here the first row labeled as ``Borovkova'' is taken from the first row labeled as ``Our approach'' in Table 3 of Borovkova et al. \cite{Borovkova}. These are basket prices that are obtained by numerically solving their non-linear equations (10), (11), (12), at page 6 of their paper. The second row labeled as ``Ours'' represents the basket prices obtained by plugging in the closed form parameters in our Lemma \ref{pp2} into the closed form approximation formula in our Theorem \ref{pp3}. The basket prices in the third row labeled ``Monte-Carlo benchmark'' are taken from the fourth row of the Table 3 of Borovkova et al.\cite{Borovkova}. 


As can be seen from Table 1, our numbers are a close match to Borovkova et al. \cite{Borovkova}'s numbers (they are all but indistinguishable in four of the six baskets at the level of four significant digits), and both are excellent matches to the Monte-Carlo benchmark. But in the cases of those baskets where the discrepancy between these two approaches is visible at the level of the 4th significant digit (Basket 3 and Basket 6), our numbers are visibly closer to the Monte-Carlo benchmark than Borovkova et al. \cite{Borovkova}'s numbers.
\begin{table}[htbp]
\centering
\caption{\centering Comparison of Borovkova's method with the benchmark}
\begin{tabular}{c|cccccc}
  \hline
  \textbf{Method} & \textbf{Basket 1} & \textbf{Basket 2} & \textbf{Basket 3} & \textbf{Basket 4} & \textbf{Basket 5} & \textbf{Basket 6} \\
  \hline
  Borovkova & 7.751 & 16.910 & 10.844 & 1.958 & 7.759  & 9.026 \\
  \hline
  Ours & 7.751 & 16.911 & 10.828 & 1.958 & 7.759  &  9.021 \\
  \hline
    \tabincell{c}{Monte-Carlo\\ benchmark} & \tabincell{c}{7.744\\ (0.014)} & \tabincell{c}{16.757\\ (0.023)} & \tabincell{c}{10.821\\ (0.018)} & \tabincell{c}{1.966\\ (0.005)} & \tabincell{c}{7.730\\ (0.01)}  &  \tabincell{c}{9.022\\ (0.015)} \\
  \hline
\end{tabular}\\
The 4th line are benchmark standard deviations
\end{table}

Next we check the performance of our Theorem \ref{p3} on time-changed models. For this, we choose a unit horizon $T=1$ and for the mixing random variable $Y=Y_{1}$ we consider three choices:   
exponential, gamma, and inverse Gaussian distributions. Specifically, using scale (rather than rate) parametrization notation, our choices are  $Y\sim \text{Exp}(1)$, $Y\sim \Gamma(2,2)$, and $Y\sim \text{IG}(1,2)$. As stated in the introduction, here $\Gamma (\cdot, \cdot)$ corresponds to the mixing distribution of the marginal law of the subordinator for  symmetric (also asymmetric) variance gamma models \cite{madan-sieneta}, $IG(\cdot, \cdot)$ corresponds to the mixing distribution of the marginal law for  inverse Gassian L\'evy processes \cite{barndorff-nielsen}. The random variable $\text{Exp}(1)$ is equal to $\Gamma (1, 1)$ in law and a NMV model with a $\text{Exp}(1)$ mixing distribution is called asymmetric Laplace distribution in the past literature, 
see \cite{kozubowski-podgorski} for its financial applications. First, we will calculate the corresponding equations (\ref{39}) for these three random variables separately below.

A few words are in order before getting into the computational details, to explain why these choices of mixing variables are relevant to mathematical finance. The use of a time-changed model acknowledges that the effective business time for a specific stock or index does not necessarily coincide with ordinary time at which trades occur. When trading activity increases per unit of time, stock local volatility will increase, and since such periods cannot be predicted except statistically, the interpretation of a stock's intrinsic time as being random, and not directly observed, is appropriate. This is consistent with the notion that volatility is not constant, and is unpredictable. Stochastic volatility models are popular ways to translate this uncertainty. They present technical challenges in continuous-time modeling in the sense that tracking a stock's underlying spot volatility is not possible in practice. Using a time-changed model instead is an intermediate modeling solution where volatility is unpredictable but need not be tracked over time. 

The correspondence between random volatility and random time changes is straightforward when a single volatility variable is used: the time scaling of Brownian motion implies immediately that $Y$ represents squared random volatility if the latter is not time dependent. For a random volatility that switches across several random values over discrete time, the same scaling implies that $Y$ would be the sum of the squares of these random volatilities. In the widespread case of Gaussian volatility which is independent of the driving Brownian motion $W$ (the case of no stochastic volatility leverage), the variable $Y$ is independent of $W$ and is distributed as the sum of squares of normals, i.e. chi-squared distributions with degrees of freedom $d=n/2$ where $n$ is the number of times the random volatility changes over the investment horizon. Our choices of exponential and Gamma distributions correspond thus to $n=2$ and $n=4$. The case of the inverse-Gamma law for $Y$ is a static version of another popular choice for stochastic volatility, as seen in \cite{Nicolas_Geoffrey}.

Returning to our calculations, for $Y_{1}\sim \text{Exp}(1)$, we have $\phi_{Y_{1}}(s)=\frac{1}{1-s}$. In this case, Equation (\ref{39}) reduces to
\begin{equation}
\begin{split}
 \frac{1}{1-\frac{9}{2}x}-\frac{3}{(1-\frac{1}{2}x)(1-2x)}+\frac{2}{(1-\frac{1}{2}x)^{3}}-|\eta_{B}(T)|(\frac{1}{1-2x}-\frac{1}{(1-\frac{1}{2}x)^{2}})^{\frac{3}{2}}=0,
\end{split}
\end{equation}
which further simplifies to
\begin{equation}
\begin{split}
 (\frac{1}{4}x^{2}+\frac{7}{4}x+\frac{15}{2})(1-2x)^{\frac{1}{2}}x^{\frac{1}{2}}=|\eta_{B}(T)|(1-\frac{9}{2}x)(1+\frac{x}{4})^{\frac{3}{2}}.
\end{split}
\end{equation}

For $Y_{1}\sim \Gamma(2,2)$, we have $\phi_{Y_{1}}(s)=(\frac{2}{2-s})^{2}$. Then Equation (\ref{39}) becomes
\begin{equation}
\begin{split}
 (\frac{2}{2-\frac{9}{2}x})^{2}-3(\frac{2}{2-\frac{1}{2}x})^{2}(\frac{2}{2-2x})^{2}+2(\frac{2}{2-\frac{1}{2}x})^{6}=|\eta_{B}(T)|((\frac{2}{2-2x})^{2}-(\frac{2}{2-\frac{1}{2}x})^{4})^{\frac{3}{2}},
\end{split}
\end{equation}
which reduces to
\begin{equation}
\begin{split}
 &(4(2-\frac{1}{2}x)^{6}(1-x)^{2}-12(2-\frac{9}{2}x)^{2}(2-\frac{1}{2}x)^{4}+128(2-\frac{9}{2}x)^{2}(1-x)^{2})(1-x)\\
 =&|\eta_{B}(T)|(2-\frac{9}{2}x)^{2}((2-\frac{1}{2}x)^{4}-16(1-x)^{2})^{\frac{3}{2}}.
\end{split}
\end{equation}

Finally for $Y_{T}\sim \text{IG}(1,2)$, we get $\phi_{Y_{T}}(s)=e^{2(1-\sqrt{1-s})}$. In this case, Equation (\ref{39}) becomes
\begin{equation}\label{2ig}
\begin{split}
 e^{2(1-\sqrt{1-\frac{9}{2}x})}-3e^{2(1-\sqrt{1-\frac{1}{2}x})}e^{2(1-\sqrt{1-2x})}+2e^{6(1-\sqrt{1-\frac{1}{2}x})}=|\eta_{B}(T)|(e^{2(1-\sqrt{1-2x})}-e^{4(1-\sqrt{1-\frac{1}{2}x})})^{\frac{3}{2}}.
\end{split}
\end{equation}
We look for non-zero solutions of this equation. It also needs to satisfy $e^{2(1-\sqrt{1-2x})}-e^{4(1-\sqrt{1-\frac{1}{2}x})}\neq 0$ which corresponds to  $\phi_{Y_{T}}(2x)-(\phi_{Y_{T}}(\frac{1}{2}x))^{2}\neq 0$ in Lemma \ref{ppp1}. We solve $x$  numerically from
\begin{equation}
\begin{split}
 \frac{e^{2(1-\sqrt{1-\frac{9}{2}x})}-3e^{2(1-\sqrt{1-\frac{1}{2}x})}e^{2(1-\sqrt{1-2x})}+2e^{6(1-\sqrt{1-\frac{1}{2}x})}}{(e^{2(1-\sqrt{1-2x})}-e^{4(1-\sqrt{1-\frac{1}{2}x})})^{\frac{3}{2}}}=|\eta_{B}(T)|.
\end{split}
\end{equation}

We use Monte-Carlo simulations, with ten million iterations, as a benchmark in our numerical comparisons below. The choice of this many iterations is to result in Monte-Carlo uncertainties on the benchmark price whose magnitude is comfortably lower than the discrepancy between the Monte-Carlo's mean and the numerical method being tested. The Monte-Carlo uncertainties reported in Tables 3 through 6 comply well with this constraint, being two or three orders of magnitude smaller than the discrepancies, allowing us to report the latter comfortably as percentages with two significant digits. 


Table 2 provides the model parameters for six baskets whose prices we compute using our approximate formula with time-changed models. These scenarios include four baskets with two stocks and two baskets with three stocks. Note that these six Scenarios have appeared in
the papers Wu et al.\cite{Wu_Diao2019} (Scenarios 1-3)  and Leccadito et al.\cite{Leccadito} (Scenarios 4-6). In particular, the Strike line in Table 2 refers to moneyness of the basket $K/B(0)$ where the strike price is $K$ and $B(0)$ is the basket's initial price, for the first three scenarios, and it refers simply to the strike price $K$ for the last three. This is in accordance with the notation in the aforementioned two references.
\begin{table}[htbp]
\footnotesize
\centering
\caption{Different Basket Scenarios}
\begin{tabular}{c|cccccc}
  \hline
   & \textbf{Scenario 1} & \textbf{Scenario 2} & \textbf{Scenario 3} & \textbf{Scenario 4} & \textbf{Scenario 5} & \textbf{Scenario 6} \\
  \hline
  Initial Price & [100,120] & [150,100] & [110,90] & [200,50] & [95,90,105]  & [100,90,95] \\

  Volatilities & [0.2,0.3] & [0.3,0.2] & [0.3,0.2] & [0.1,0.15] & [0.2,0.3,0.25]  &  [0.25,0.3,0.2] \\

  Weights & [-1,1] & [-1,1] & [0.7,0.3] & [-1,1] & [1,-0.8,-0.5]  &  [0.6,0.8,-1] \\

  Correlations & $\rho_{12}=0.9$ & $\rho_{12}=0.3$ & $\rho_{12}=0.9$ & $\rho_{12}=0.8$ & \multicolumn{2}{c}{\tabincell{c}{$\rho_{12}=\rho_{23}=0.9$\\ $\rho_{13}=0.8$}} \\

  Strike & \multicolumn{3}{c}{Moneyness $\in\{0.8,0.9,1.0,1.1,1.2\}$} & $K=-140$ & $K=-30$  & $K=35$  \\

  Mixing & \multicolumn{6}{c}{$Y=\text{Exp}(1),\quad Y=\Gamma(2,2),\quad Y=\text{IG}(1,2)$} \\
  \hline
\end{tabular}\\
\begin{tabular}{l}
Notes: The risk-less interest rate $r$ equals to $0.03$. The initial price of the underlying assets are listed\\
in the first row. The volatilities, weights, and the coefficients of correlations of the assets are placed\\
in the second, third, forth row, respectively. The fifth row is the moneyness $\frac{K}{B(0)}$ or the strike price of\\
the basket options. The sixth row lies the different Mixing random variables.\\
\hline
\end{tabular}
\end{table}

To clarify the parameters defined in the Strike line of this table, for example, in Scenario 1, moneyness equals to 0.8 means the strike price $K=0.8\times B(0)=0.8\times(-100+120)=16$, moneyness equals to 0.9 means the strike price $K=0.9\times B(0)=18$, and so on. For Scenarios 2 and 3, the strike prices are computed using the same operation as in Scenario 1. The strike prices are given in the table in Scenarios 4 to 6. We chose these two different parametrizations to respect the notation in the two articles from whence these scenarios came.

We use two measures of errors,  the percentage of "good prices" ($C_{1}$) and the mean absolute percentage error ($C_{2}$), in our numerical comparisons as was done in the two papers  Leccadito et al.\cite{Leccadito} and Wu et al.\cite{Wu_Diao2019}. These measures are calculated as follows. Suppose there are $L$ (positive integer) cases in which one would like to check the accuracy of our formula. For instance, for scenarios 1, 2, and 3, since there are five moneynesses to check for each mixing distribution, that results in $L=5$ cases for each $Y$ in each scenario. For scenarios 4, 5, and 6, since they have single strike prices, we group each of the three scenarios by mixing distribution, resulting in $L=3$ for all three scenarios collectively. For each case $\ell =1,\ldots,L$, let $Val_{\ell}$ denote the outcome of our computational method and  $MC_{\ell}$ be the outcome of Monte-Carlo method for the same case. Then
\begin{equation}
\begin{split}
 C_{1}&=\frac{\sum_{\ell=1}^{L}I(|\frac{Val_{\ell}-MC_{\ell}}{MC_{\ell}}|<2\%)}{L}\times100\%,\\
 C_{2}&=\frac{\sum_{\ell=1}^{L}|\frac{Val_{\ell}-MC_{\ell}}{MC_{\ell}}|}{L}\times100\%,
\end{split}
\end{equation}
where $I(A)$ is the indicator function of the event $A$, equal to $1$ is $A$ is true, and zero otherwise.

From the above definitions of $C_{1}$ and $C_{2}$, we can see that one wants a $C_{1}$ which is a close to 100\% as possible, and a $C_{2}$ which is as close to 0\% as possible. The values of these $C_1,C_2$ are collected in Tables 3 thru 6 below. As can be seen there, our Theorem \ref{p3}  gives highly accurate basket prices. The $C_1$'s indicates that $100\%$ of our approximations are within 2 percentage points of the benchmark values, and the $C_2$'s show that our mean absolute percentage error sizes are on the order of one percentage point or slightly less, typically, with the accuracy improving to less than a tenth of a percentage point in Scenario 3 for all mixing distributions. We note that Scenario 3 is the only case with exclusively positive basket weights. We suspect that our formulas perform better in such cases. The overall mean absolute percentage of error of our method for all cases considered in Tables 3 thru 6 is 0.56\%. These numerical tests are encouraging for the accuracy of our formulas in Theorem  \ref{p3}, since this is consistent across many cases with various parameters, weights, and mixing distributions.

\begin{table}[htbp]
\footnotesize
\centering
\caption{The performance of our method, scenario 1}
\begin{tabular}{c|cc|cc|cc}
  \hline
  & \multicolumn{2}{c}{$\bm{Y=Exp(1)}$} & \multicolumn{2}{|c}{$\bm{Y=\Gamma(2,2)}$} & \multicolumn{2}{|c}{ $\bm{Y=IG(1,2)}$}\\
  \hline  
  & \textbf{MC}  & \textbf{Ours} & \textbf{MC}  & \textbf{Ours} & \textbf{MC}  & \textbf{Ours} \\

  0.8 & 9.3540(0.0064)  & 9.4214(0.72\%)  & 9.7012(0.0057)  & 9.7275(0.27\%)  & 9.7601(0.0057)  & 9.8083(0.49\%)  \\

  0.9 & 8.3827(0.0062)  & 8.4529(0.84\%)  & 8.7296(0.0055)  & 8.7581(0.33\%)  & 8.7898(0.0056)  & 8.8378(0.55\%)  \\

  1.0 & 7.5417(0.0061)  & 7.6117(0.93\%)  & 7.8562(0.0054)  & 7.8858(0.38\%)  & 7.9112(0.0054)  & 7.9579(0.59\%) \\

  1.1 & 6.8105(0.0059)  & 6.8780(0.99\%)  & 7.0747(0.0052)  & 7.1043(0.42\%)  & 7.1194(0.0052)  & 7.1639(0.62\%) \\

  1.2 & 6.1717(0.0058)  & 6.2353(1.03\%)  & 6.3771(0.0051)  & 6.4060(0.45\%)  & 6.4085(0.0051)  & 6.4502(0.65\%) \\
  \hline
  $C_{1}$ & \multicolumn{2}{c}{100.00\%}  & \multicolumn{2}{|c}{100.00\%} & \multicolumn{2}{|c}{100.00\%} \\

  $C_{2}$ & \multicolumn{2}{c}{0.90\%}    &  \multicolumn{2}{|c}{0.37\%}   &  \multicolumn{2}{|c}{0.58\%}  \\
  \hline
%
%

\end{tabular}
\begin{tabular}{l}
Notes: Parameters are given in Table 1. The results of Monte Carlo with the standard error in the bracket\\
are listed in the $MC$ column. The results of our method  and the absolute percentage error are displayed\\
in the $Ours$ column. Row $C_{1}$ is the percentage of "good price" and row $C_{2}$ represents the mean absolute\\
percentage error.\\
\hline
\end{tabular}
\end{table}

\begin{table}[htbp]
\footnotesize
\centering
\caption{The performance of our method, scenario 2}
\begin{tabular}{c|cc|cc|cc}
  \hline
  & \multicolumn{2}{c}{$\bm{Y=Exp(1)}$} & \multicolumn{2}{|c}{$\bm{Y=\Gamma(2,2)}$} & \multicolumn{2}{|c}{ $\bm{Y=IG(1,2)}$}\\
  \hline
  & \textbf{MC}  & \textbf{Ours} & \textbf{MC}  & \textbf{Ours} & \textbf{MC}  & \textbf{Ours} \\

  0.8 & 10.1565(0.0061)  & 10.1627(0.06\%)  & 10.8574(0.0060)  & 10.9906(1.23\%)  & 11.0131(0.0059)  & 11.1013(0.80\%)  \\

  0.9 & 12.2973(0.0066)  & 12.3898(0.75\%)  & 13.0688(0.0065)  & 13.2499(1.39\%)  & 13.2423(0.0064)  & 13.3770(1.02\%)  \\

  1.0 & 14.8167(0.0070)  & 14.9907(1.17\%)  & 15.5660(0.0070)  & 15.7861(1.41\%)  & 15.7384(0.0070)  & 15.9116(1.10\%) \\

  1.1 & 17.6883(0.0075)  & 17.9198(1.31\%)  & 18.3386(0.0074)  & 18.5865(1.35\%)  & 18.4918(0.0075)  & 18.6949(1.10\%) \\

  1.2 & 20.8524(0.0079)  & 21.1214(1.29\%)  & 21.3661(0.0079)  & 21.6310(1.24\%)  & 21.4880(0.0079)  & 21.7121(1.04\%) \\
  \hline
  $C_{1}$ & \multicolumn{2}{c}{100.00\%}  & \multicolumn{2}{|c}{100.00\%} & \multicolumn{2}{|c}{100.00\%} \\

  $C_{2}$ & \multicolumn{2}{c}{0.92\%}    &  \multicolumn{2}{|c}{1.32\%}   &  \multicolumn{2}{|c}{1.01\%}  \\
  \hline

\end{tabular}
\begin{tabular}{l}
Notes: Parameters are given in Table 1. The results of Monte Carlo with the standard error in the bracket are\\
listed in the $MC$ column. The results of our method  and the absolute percentage error are displayed in the $Ours$\\
column. Row $C_{1}$ is the percentage of "good price" and row $C_{2}$ represents the mean absolute percentage error.\\
\hline
\end{tabular}
\end{table}

\begin{table}[htbp]
\footnotesize
\centering
\caption{The performance of our method, scenario 3}
\begin{tabular}{c|cc|cc|cc}
  \hline
  & \multicolumn{2}{c}{$\bm{Y=Exp(1)}$} & \multicolumn{2}{|c}{$\bm{Y=\Gamma(2,2)}$} & \multicolumn{2}{|c}{ $\bm{Y=IG(1,2)}$}\\
  \hline
  & \textbf{MC}  & \textbf{Ours} & \textbf{MC}  & \textbf{Ours} & \textbf{MC}  & \textbf{Ours} \\

  0.8 & 25.2992(0.0090)  & 25.2967(0.01\%)  & 25.4051(0.0086)  & 25.3848(0.08\%)  & 25.3672(0.0086)  & 25.3714(0.02\%)  \\

  0.9 & 17.4806(0.0085)  & 17.4779(0.02\%)  & 17.8465(0.0079)  & 17.8327(0.08\%)  & 17.8799(0.0079)  & 17.8857(0.03\%)  \\

  1.0 & 11.4667(0.0078)  & 11.4657(0.01\%)  & 12.0070(0.0071)  & 11.9987(0.07\%)  & 12.0898(0.0071)  & 12.0973(0.06\%) \\

  1.1 & 7.6897(0.0070)  & 7.6919(0.03\%)  & 7.9797(0.0062)  & 7.9744(0.07\%)  & 8.0080(0.0062)  & 8.0186(0.13\%) \\

  1.2 & 5.3455(0.0062)  & 5.3512(0.11\%)  & 5.3472(0.0054)  & 5.3437(0.07\%)  & 5.3073(0.0054)  & 5.3188(0.22\%) \\
  \hline
  $C_{1}$ & \multicolumn{2}{c}{100.00\%}  & \multicolumn{2}{|c}{100.00\%} & \multicolumn{2}{|c}{100.00\%} \\

  $C_{2}$ & \multicolumn{2}{c}{0.03\%}    &  \multicolumn{2}{|c}{0.07\%}   &  \multicolumn{2}{|c}{0.09\%}  \\
  \hline

\end{tabular}
\begin{tabular}{l}
Notes: Parameters are given in Table 1. The results of Monte Carlo with the standard error in the bracket are\\
listed in the $MC$ column. The results of our method  and the absolute percentage error are displayed in the $Ours$\\
column. Row $C_{1}$ is the percentage of "good price" and row $C_{2}$ represents the mean absolute percentage error.\\
\hline
\end{tabular}
\end{table}

\begin{table}[htbp]
\footnotesize
\centering
\caption{The performance of our method, scenario 4-6}
\begin{tabular}{c|cc|cc|cc}
  \hline
  & \multicolumn{2}{c}{$\bm{Y=Exp(1)}$} & \multicolumn{2}{|c}{$\bm{Y=\Gamma(2,2)}$} & \multicolumn{2}{|c}{ $\bm{Y=IG(1,2)}$}\\
  \hline
  & \textbf{MC}  & \textbf{Ours} & \textbf{MC}  & \textbf{Ours} & \textbf{MC}  & \textbf{Ours} \\

  sce4 & 1.1595(0.0013)  & 1.1473(1.05\%)  & 1.1457(0.0012)  & 1.1438(0.16\%)  & 1.1310(0.0012)  & 1.1279(0.27\%) \\

  sce5 & 6.7895(0.0029)  & 6.8238(0.51\%)  & 7.1012(0.0029)  & 7.1307(0.42\%)  & 7.1661(0.0029)  & 7.1926(0.37\%) \\

  sce6 & 8.9799(0.0062)  & 9.0029(0.26\%)  & 9.3498(0.0056)  & 9.3764(0.28\%)  & 9.4288(0.0056)  & 9.4512(0.24\%) \\
  
  \hline
  $C_{1}$ & \multicolumn{2}{c}{100.00\%}  & \multicolumn{2}{|c}{100.00\%} & \multicolumn{2}{|c}{100.00\%} \\

  $C_{2}$ & \multicolumn{2}{c}{0.60\%}    &  \multicolumn{2}{|c}{0.29\%}   &  \multicolumn{2}{|c}{0.29\%}  \\
  \hline

\end{tabular}
\begin{tabular}{l}
Notes: Parameters are given in Table 1. The results of Monte Carlo with the standard error in the bracket\\
are listed in the $MC$ column. The results of our method  and the absolute percentage error are displayed\\
in the $Ours$ column. Row $C_{1}$ is the percentage of "good price" and row $C_{2}$ represents the mean absolute\\
percentage error.\\
\hline
\end{tabular}
\end{table}

\newpage

\section{Conclusions}
In this paper, we first give a closed-form
solution for the system of non-linear equations in Borovkova et al.\cite{Borovkova} that needs to be solved to obtain approximate values for basket option prices. To do this, we first transform this system of non-linear equations into a unary cubic equation and then apply the Cardano's formula to get the single real root of this unary cubic equation. We then use this solution to write closed-form approximate basket prices in an explicit  functional form, depending on the mean, variance, and skewness of the basket. This closed form approximation to the basket option price made possible by our approach, is much quicker and more efficient than the computational method proposed in Borovkova et al.\cite{Borovkova} for the same option pricing approximation. Both methods are based on the same moment-matching mathematical methodology. Numerics we perform via our method and the Monte-Carlo method as a benchmark, on a selection of secenarios, show that our formula agrees with Borovkova et al.'s to a very high degree, and that these prices are indeed excellent approximations to the true basket option prices.
 
Our second contribution is to elucidate the problem of pricing basket options when the underlying assets are modeled via time-changing Brownian motion by subordinators. We apply the matching of the first three moments approach, as in the case of log-normal models, for this class of time-changed models also. As a result, we  provide a closed form approximation formula for basket option prices for these time-changed models. Since we only match the first three moments, we are able to reduce the problem of solving the associated system of non-linear equations further and reduce it eventually to finding the root of a real valued function on the real line. Even though we only match the first three moments, our numerical tests show that the corresponding closed form approximation formula obtained performs very accurately, on a wide range of model parameters and time-change distributions, with error discrepancies on the order of 1\% or less.

\section{Appendix: Proofs}

\textbf{Proof of Lemma \ref{pp2}}: We divide into cases.

Case 1: If $\eta_{B}(T)>0$, then $c=1$ and
\begin{equation}\label{577}
\begin{split}
 &M_{1}=e^{\frac{1}{2}s^{2}+m}+\tau,\\
 &M_{2}=e^{2s^{2}+2m}+2\tau e^{\frac{1}{2}s^{2}+m}+\tau^{2},\\
 &M_{3}=e^{\frac{9}{2}s^{2}+3m}+3\tau e^{2s^{2}+2m}+3\tau^{2}e^{\frac{1}{2}s^{2}+m}+\tau^{3}.
\end{split}
\end{equation}
From the first two equations of (\ref{577}), we  obtain
\begin{equation}\label{ab1}
\begin{split}
 M_{2}-M_{1}^{2}=e^{2s^{2}+2m}-e^{s^{2}+2m}.
\end{split}
\end{equation}
From the last two equations of (\ref{577}), we obtain
\begin{equation}\label{599}
\begin{split}
 M_{3}-M_{1}^{3}&=3\tau(e^{2s^{2}+2m}-e^{s^{2}+2m})+e^{\frac{9}{2}s^{2}+3m}-e^{\frac{3}{2}s^{2}+3m}\\
 &=3(M_{1}-e^{\frac{1}{2}s^{2}+m})(M_{2}-M_{1}^{2})+e^{\frac{9}{2}s^{2}+3m}-e^{\frac{3}{2}s^{2}+3m}\\
 &=3M_{1}M_{2}-3M_{1}^{3}-3(M_{2}-M_{1}^{2})e^{\frac{1}{2}s^{2}+m}+e^{\frac{9}{2}s^{2}+3m}-e^{\frac{3}{2}s^{2}+3m}.
\end{split}
\end{equation}
From (\ref{577}) and (\ref{599}) we obtain
\begin{equation}\label{ab2}
\begin{split}
 M_{3}-3M_{1}M_{2}+2M_{1}^{3}=e^{\frac{9}{2}s^{2}+3m}-e^{\frac{3}{2}s^{2}+3m}-3(M_{2}-M_{1}^{2})e^{\frac{1}{2}s^{2}+m}.
\end{split}
\end{equation}
Observe that the left-hand-side of (\ref{ab1}) 
is the variance of the basket $B(T)$ and the left-hand-side of  (\ref{ab2}) is the third central moment of the basket $B(T)$.

From (\ref{ab1}) we obtain
\begin{equation}\label{ab3}
\begin{split}
 e^{m}=\sqrt{\frac{M_{2}-M_{1}^{2}}{e^{2s^{2}}-e^{s^{2}}}}.
\end{split}
\end{equation}
We plug (\ref{ab3}) into (\ref{ab2}) and obtain
\begin{equation}\label{ab4}
\begin{split}
 M_{3}-3M_{1}M_{2}+2M_{1}^{3}=(\frac{M_{2}-M_{1}^{2}}{e^{2s^{2}}-e^{s^{2}}})^{\frac{3}{2}}(e^{\frac{9}{2}s^{2}}-e^{\frac{3}{2}s^{2}})-3(M_{2}-M_{1}^{2})
 (\frac{M_{2}-M_{1}^{2}}{e^{2s^{2}}-e^{s^{2}}})^{\frac{1}{2}}e^{\frac{1}{2}s^{2}}.
\end{split}
\end{equation}
If we let $x=e^{s^{2}}$ in (\ref{ab4}), the equation (\ref{ab4}) becomes
\begin{equation}\label{633}
\begin{split}
 M_{3}-3M_{1}M_{2}+2M_{1}^{3}=(\frac{M_{2}-M_{1}^{2}}{x^{2}-x})^{\frac{3}{2}}(x^{\frac{9}{2}}-x^{\frac{3}{2}})-3(M_{2}-M_{1}^{2})(\frac{M_{2}-M_{1}^{2}}{x^{2}-x})^{\frac{1}{2}}x^{\frac{1}{2}}.
\end{split}
\end{equation}
Observe here that $s\neq0$ and therefore $x>1$. The formula (\ref{633}) simplifies to
\begin{equation}
\begin{split}
 (x-1)^{\frac{1}{2}}(x+2)=\frac{M_{3}-3M_{1}M_{2}+2M_{1}^{3}}{(M_{2}-M_{1}^{2})^{\frac{3}{2}}}=\eta_{B}(T),
\end{split}
\end{equation}
which means $x$ is the real root of the unary cubic equation $x^{3}+3x^2-4-\eta_{B}^{2}(T)=0$.

To solve this cubic equation we let $y=x+1$ and reduce it to the new cubic equation $y^{3}-3y-2-\eta_{B}^{2}(T)=0$.
The discriminant of this new cubic equation is $\Delta=(1+\frac{\eta_{B}^{2}(T)}{2})^{2}-1>0$. This means that the equation should have only one real root and according to the Cardano's formula, the real root is
\begin{equation}
\begin{split}
 y=\sqrt[3]{1+\frac{1}{2}\eta_{B}^{2}(T)+\eta_{B}(T)\sqrt{1+\frac{1}{4}\eta_{B}^{2}(T)}}+\sqrt[3]{1+\frac{1}{2}\eta_{B}^{2}(T)-\eta_{B}(T)\sqrt{1+\frac{1}{4}\eta_{B}^{2}(T)}}.
\end{split}
\end{equation}
Therefore  the only real root of the initial cubic equation $x^{3}+3x^2-4-\eta_{B}^{2}(T)=0$ is
\begin{equation}\label{xx}
\begin{split}
 x=\sqrt[3]{1+\frac{1}{2}\eta_{B}^{2}(T)+\eta_{B}(T)\sqrt{1+\frac{1}{4}\eta_{B}^{2}(T)}}+\sqrt[3]{1+\frac{1}{2}\eta_{B}^{2}(T)-\eta_{B}(T)\sqrt{1+\frac{1}{4}\eta_{B}^{2}(T)}}-1.
\end{split}
\end{equation}
From $x=e^{s^2}$, we have
\begin{equation}
\begin{split}
 s=(\ln(x))^{\frac{1}{2}}.
\end{split}
\end{equation}
From (\ref{ab3}), we obtain
\begin{equation}
\begin{split}
 m=\frac{1}{2}\ln(\frac{M_{2}-M_{1}^{2}}{x(x-1)})=\frac{1}{2}\ln(\frac{\sigma^{2}_{B}(T)}{x(x-1)}).
\end{split}
\end{equation}
The parameter $\tau$ can be calculated as
\begin{equation}
\begin{split}
 \tau=M_{1}-e^{\frac{1}{2}s^{2}+m}=M_{1}-x^{\frac{1}{2}}\frac{\sigma_{B}(T)}{(x(x-1))^{\frac{1}{2}}}=\mu_{B}(T)-\frac{\sigma_{B}(T)}{(x-1)^{\frac{1}{2}}}.
\end{split}
\end{equation}

Case 2: If $\eta_{B}(T)<0$, then $c=-1$ and in this case we have
\begin{equation}
\begin{split}
 &M_{1}=-e^{\frac{1}{2}s^{2}+m}-\tau,\\
 &M_{2}=e^{2s^{2}+2m}+2\tau e^{\frac{1}{2}s^{2}+m}+\tau^{2},\\
 &M_{3}=-e^{\frac{9}{2}s^{2}+3m}-3\tau e^{2s^{2}+2m}-3\tau^{2}e^{\frac{1}{2}s^{2}+m}-\tau^{3}.
\end{split}
\end{equation}
By following the parallel calculations as in the Case 1 above we  obtain the same cubic equation $x^{3}+3x^2-4-\eta_{B}^{2}(T)=0$ and this leads us to the solution (\ref{xx}).

\vspace{0.1in}

\textbf{Proof of Theorem \ref{pp3}:} We divide into cases.

Case 1: If $c=1,K\leq\tau$, (\ref{aa1}) becomes
\begin{equation}
\begin{split}
 \hat{\Pi}&= e^{-rT}E[\left(e^{sN+m}+\tau-K\right)]\\
 &=e^{-rT}(e^{m}E[e^{sN}]+\tau-K)\\
 &=e^{-rT}(e^{m+\frac{1}{2}s^{2}}+\tau-K).
\end{split}
\end{equation}

Case 2: If $c=1,K>\tau$, (\ref{aa1}) becomes
\begin{equation}\label{ac2}
\begin{split}
 \hat{\Pi}&= e^{-rT}E[\left(e^{sN+m}+\tau-K\right)^{+}]\\
 &=e^{-rT}E[\left(e^{sN+m}+\tau-K\right)\cdot I(e^{sN+m}+\tau-K>0)]\\
 &=e^{-rT+m}E[e^{sN}\cdot I(e^{sN+m}+\tau-K>0)]-e^{-rT}(K-\tau)P(e^{sN+m}+\tau-K>0),
\end{split}
\end{equation}
where $I(\cdot)$ is the indicator function. We have
\begin{equation}\label{ac3}
\begin{split}
 E[e^{sN}\cdot I(e^{sN+m}+\tau-K>0)]&=e^{\frac{1}{2}s^{2}}P(e^{sN+s^{2}+m}+\tau-K>0)\\
 &=e^{\frac{1}{2}s^{2}}\Phi(\frac{-\ln(K-\tau)+m+s^{2}}{s}),
\end{split}
\end{equation}
and
\begin{equation}\label{ac4}
\begin{split}
 P(e^{sN+m}+\tau-K>0)=\Phi(\frac{-\ln(K-\tau)+m}{s}).
\end{split}
\end{equation}
We plug (\ref{ac3}) and (\ref{ac4})  into (\ref{ac2}) and obtain 
\begin{equation}
\begin{split}
 \hat{\Pi}&= e^{-rT+m+\frac{1}{2}s^{2}}\Phi(\frac{-\ln(K-\tau)+m+s^{2}}{s})-e^{-rT}(K-\tau)\Phi(\frac{-\ln(K-\tau)+m}{s}).
\end{split}
\end{equation}

Case 3: If $c=-1,K\geq-\tau$, (\ref{aa1}) becomes
\begin{equation}
\begin{split}
 \hat{\Pi}= e^{-rT}E[\left(-e^{sN+m}-\tau-K\right)^{+}]=0.
\end{split}
\end{equation}

Case 4: If $c=-1,K<-\tau$, (\ref{aa1}) becomes
\begin{equation}\label{ac5}
\begin{split}
 \hat{\Pi}&= e^{-rT}E[\left(-e^{sN+m}-\tau-K\right)^{+}]\\
 &=e^{-rT}E[\left(-e^{sN+m}-\tau-K\right)\cdot I(-e^{sN+m}-\tau-K>0)]\\
 &=-e^{-rT+m}E[e^{sN}\cdot I(-e^{sN+m}-\tau-K>0)]\\
 &+e^{-rT}(-K-\tau)P(-e^{sN+m}-\tau-K>0),
\end{split}
\end{equation}
where $I(\cdot)$ is the indicator function. We have
\begin{equation}\label{ac6}
\begin{split}
 E[e^{sN}\cdot I(-e^{sN+m}-\tau-K>0)]&=e^{\frac{1}{2}s^{2}}P(-e^{sN+s^{2}+m}-\tau-K>0)\\
 &=e^{\frac{1}{2}s^{2}}\Phi(\frac{\ln(-K-\tau)-m-s^{2}}{s}),
\end{split}
\end{equation}
and
\begin{equation}\label{ac7}
\begin{split}
 P(-e^{sN+m}-\tau-K>0)=\Phi(\frac{\ln(-K-\tau)-m}{s}).
\end{split}
\end{equation}
Plugging (\ref{ac6}) and (\ref{ac7}) back into (\ref{ac5}), we obtain
\begin{equation}
\begin{split}
 \hat{\Pi}&=-e^{-rT+m+\frac{1}{2}s^{2}}\Phi(\frac{\ln(-K-\tau)-m-s^{2}}{s})+e^{-rT}(-K-\tau)\Phi(\frac{\ln(-K-\tau)-m}{s}).
\end{split}
\end{equation}
Finally, we let 
\begin{equation}
\begin{split}
 d_{11}&=\frac{-\ln(K-\tau)+m+s^{2}}{s},\\
 d_{12}&=\frac{-\ln(K-\tau)+m}{s},\\
 d_{21}&=\frac{\ln(-K-\tau)-m-s^{2}}{s},\\
 d_{22}&=\frac{\ln(-K-\tau)-m}{s},
\end{split}
\end{equation}
and obtain  (\ref{pp3_1}).

\vspace{0.1in}

\textbf{Proof of Proposition \ref{th1}:}
We divide into cases.

Case 1: When $c=1,K\leq\tau$, we have
\begin{equation}
\begin{split}
 \hat{\Pi}&= e^{-rT}(e^{m+\frac{1}{2}s^{2}}+\tau-K)\\
 &=e^{-rT}(\frac{\sigma_{B}(T)}{\sqrt{x(x-1)}}\sqrt{x}+\mu_B(T)-\frac{\sigma_{B}(T)}{\sqrt{x-1}}-K)\\
 &=e^{-rT}(\mu_B(T)-K).
\end{split}
\end{equation}
Therefore
\begin{equation}
\begin{split}
 \frac{\partial\hat{\Pi}}{\partial \mu_B(T)}=e^{-rT},\quad
 \frac{\partial\hat{\Pi}}{\partial\sigma_{B}(T)}=0,\quad
 \frac{\partial\hat{\Pi}}{\partial\eta_{B}(T)}=0.
\end{split}
\end{equation}

Case 2: When $c=1,K>\tau$, we have
\begin{equation}
\begin{split}
 \hat{\Pi}&= e^{-rT+m+\frac{1}{2}s^{2}}\Phi(d_{11})-e^{-rT}(K-\tau)\Phi(d_{12})\\
 &=e^{-rT}\frac{\sigma_{B}(T)}{\sqrt{x(x-1)}}\sqrt{x}\Phi(d_{11})-e^{-rT}(K-\mu_{B}(T)+\frac{\sigma_{B}(T)}{\sqrt{x-1}})\Phi(d_{12})\\
 &=e^{-rT}\frac{\sigma_{B}(T)}{\sqrt{x-1}}\Phi(d_{11})-e^{-rT}(K-\mu_{B}(T)+\frac{\sigma_{B}(T)}{\sqrt{x-1}})\Phi(d_{12}),
\end{split}
\end{equation}
where
\begin{equation}
\begin{split}
 d_{11}&=\frac{-\ln(K-\tau)+m+s^{2}}{s}\\
 &=\frac{-\ln(K-\mu_{B}(T)+\frac{\sigma_{B}(T)}{\sqrt{x-1}})+\ln(\sigma_{B}(T))-\frac{1}{2}\ln(x(x-1))}{(\ln(x))^{\frac{1}{2}}}+(\ln(x))^{\frac{1}{2}},\\
 d_{12}&=\frac{-\ln(K-\tau)+m}{s}\\
 &=\frac{-\ln(K-\mu_{B}(T)+\frac{\sigma_{B}(T)}{\sqrt{x-1}})+\ln(\sigma_{B}(T))-\frac{1}{2}\ln(x(x-1))}{(\ln(x))^{\frac{1}{2}}}.
\end{split}
\end{equation}
We have the following  relation between $d_{11}$ and $d_{12}$, which we will apply several times in the proof,
\begin{equation}\label{new2}
\begin{split}
 \phi(d_{11})&=\frac{1}{\sqrt{2\pi}}e^{-\frac{d_{11}^{2}}{2}}\\
 &=\frac{1}{\sqrt{2\pi}}e^{-\frac{d_{12}^{2}}{2}}e^{-d_{12}(\ln(x))^{\frac{1}{2}}-\frac{1}{2}\ln(x)}\\
 &=\phi(d_{12})e^{\ln(K-\mu_{B}(T)+\frac{\sigma_{B}(T)}{\sqrt{x-1}})-\ln(\sigma_{B}(T))+\frac{1}{2}\ln(x-1)}\\
 &=\phi(d_{12})(K-\mu_{B}(T)+\frac{\sigma_{B}(T)}{\sqrt{x-1}})\frac{\sqrt{x-1}}{\sigma_{B}(T)}.
\end{split}
\end{equation}
Since
\begin{equation}
\begin{split}
 \frac{\partial d_{11}}{\partial \mu_{B}(T)}=\frac{\partial d_{12}}{\partial \mu_{B}(T)}=\frac{1}{(K-\mu_{B}(T)+\frac{\sigma_{B}(T)}{\sqrt{x-1}})(\ln(x))^{\frac{1}{2}}}
\end{split}
\end{equation}
we have
\begin{equation}\label{888}
\begin{split}
 \frac{\partial\hat{\Pi}}{\partial \mu_{B}(T)}&=e^{-rT}\frac{\sigma_{B}(T)}{\sqrt{x-1}}\phi(d_{11})\frac{\partial d_{11}}{\partial \mu_{B}(T)}\\
 &+e^{-rT}\Phi(d_{12})-e^{-rT}(K-\mu_{B}(T)+\frac{\sigma_{B}(T)}{\sqrt{x-1}})\phi(d_{12})\frac{\partial d_{12}}{\partial \mu_{B}(T)}\\
 &=e^{-rT}\Phi(d_{12}),
\end{split}
\end{equation}
where we plugged in the relation (\ref{new2}) to get the second equation of (\ref{888}). Next, since
\begin{equation}
\begin{split}
 \frac{\partial d_{11}}{\partial\sigma_{B}(T)}=\frac{\partial d_{12}}{\partial\sigma_{B}(T)}=\frac{\sqrt{x-1}(K-\mu_{B}(T))}{\sigma_{B}(T)(\sigma_{B}(T)+\sqrt{x-1}(K-\mu_{B}(T)))(\ln(x))^{\frac{1}{2}}},
\end{split}
\end{equation}
we have 
\begin{equation}
\begin{split}
 \frac{\partial\hat{\Pi}}{\partial\sigma_{B}(T)}&=\frac{e^{-rT}}{\sqrt{x-1}}(\Phi(d_{11})-\Phi(d_{12}))+e^{-rT}\frac{\sigma_{B}(T)}{\sqrt{x-1}}\phi(d_{11})\frac{\partial d_{11}}{\partial\sigma_{B}(T)}\\
 &-e^{-rT}(K-\mu_{B}(T)+\frac{\sigma_{B}(T)}{\sqrt{x-1}})\phi(d_{12})\frac{\partial d_{12}}{\partial\sigma_{B}(T)}\\
 &=\frac{e^{-rT}}{\sqrt{x-1}}(\Phi(d_{11})-\Phi(d_{12})),
\end{split}
\end{equation}
where we applied the equation (\ref{new2}) to get the last equation as well. To get the  derivative with respect to the  skewness, we first calculate the following
\begin{equation}
\begin{split}
 \frac{\partial d_{11}}{\partial x}&=\frac{\partial d_{12}}{\partial x}+\frac{1}{2x(\ln(x))^{\frac{1}{2}}}\\
 &=\frac{1}{2(\ln(x))^{\frac{1}{2}}(x-1)}(\frac{\sigma_{B}(T)}{(K-\mu_{B}(T)+\frac{\sigma_{B}(T)}{\sqrt{x-1}})\sqrt{x-1}}-\frac{2x-1}{x})\\
 &+\frac{1}{2x(\ln(x))^{\frac{3}{2}}}(\ln(K-\mu_{B}(T)+\frac{\sigma_{B}(T)}{\sqrt{x-1}})-\ln(\sigma_{B}(T))+\frac{1}{2}\ln(x(x-1)))+\frac{1}{2x(\ln(x))^{\frac{1}{2}}}.
\end{split}
\end{equation}
Finally we have
\begin{equation}
\begin{split}
 \frac{\partial\hat{\Pi}}{\partial\eta_{B}(T)}&=-\frac{1}{2}e^{-rT}\frac{\sigma_{B}(T)}{(x-1)^{\frac{3}{2}}}(\Phi(d_{11})-\Phi(d_{12}))\frac{\partial x}{\partial\eta_{B}(T)}\\
 &+e^{-rT}\frac{\sigma_{B}(T)}{\sqrt{x-1}}\phi(d_{11})\frac{\partial d_{11}}{\partial x}\frac{\partial x}{\partial\eta_{B}(T)}-e^{-rT}(K-\mu_{B}(T)+\frac{\sigma_{B}(T)}{\sqrt{x-1}})\phi(d_{12})\frac{\partial d_{12}}{\partial x}\frac{\partial x}{\partial\eta_{B}(T)}\\
 &=-\frac{1}{2}e^{-rT}\frac{\sigma_{B}(T)}{(x-1)^{\frac{3}{2}}}(\Phi(d_{11})-\Phi(d_{12}))\frac{\partial x}{\partial\eta_{B}(T)}+\frac{e^{-rT}\sigma_{B}(T)}{2x\sqrt{(x-1)\ln(x)}}\phi(d_{11})\frac{\partial x}{\partial\eta_{B}(T)}\\
 &+e^{-rT}\frac{\sigma_{B}(T)}{\sqrt{x-1}}\phi(d_{11})\frac{\partial d_{12}}{\partial x}\frac{\partial x}{\partial\eta_{B}(T)}-e^{-rT}(K-\mu_{B}(T)+\frac{\sigma_{B}(T)}{\sqrt{x-1}})\phi(d_{12})\frac{\partial d_{12}}{\partial x}\frac{\partial x}{\partial\eta_{B}(T)}\\
 &=(-\frac{1}{2}e^{-rT}\frac{\sigma_{B}(T)}{(x-1)^{\frac{3}{2}}}(\Phi(d_{11})-\Phi(d_{12}))+\frac{e^{-rT}\sigma_{B}(T)}{2x\sqrt{(x-1)\ln(x)}}\phi(d_{11}))\frac{\partial x}{\partial\eta_{B}(T)}\\
 &=\frac{e^{-rT}\sigma_{B}(T)}{2(x-1)^{\frac{1}{2}}}(\frac{-\Phi(d_{11})+\Phi(d_{12})}{x-1}+\frac{\phi(d_{11})}{x(\ln(x))^{\frac{1}{2}}})\frac{\partial x}{\partial\eta_{B}(T)},
\end{split}
\end{equation}
where
\begin{equation}
\begin{split}
 \frac{\partial x}{\partial\eta_{B}(T)}&=\frac{1}{3}(1+\frac{1}{2}\eta_{B}^{2}(T)+\eta_{B}(T)(1+\frac{1}{4}\eta_{B}^{2}(T))^{\frac{1}{2}})^{-\frac{2}{3}}(\eta_{B}(T)+(1+\frac{1}{2}\eta_{B}^{2}(T))(1+\frac{1}{4}\eta_{B}^{2}(T))^{-\frac{1}{2}})\\
 &+\frac{1}{3}(1+\frac{1}{2}\eta_{B}^{2}(T)-\eta_{B}(T)(1+\frac{1}{4}\eta_{B}^{2}(T))^{\frac{1}{2}})^{-\frac{2}{3}}(\eta_{B}(T)-(1+\frac{1}{2}\eta_{B}^{2}(T))(1+\frac{1}{4}\eta_{B}^{2}(T))^{-\frac{1}{2}}).
\end{split}
\end{equation}

Case 3 : .When $c=-1,K\geq-\tau$, we have
\begin{equation}
\begin{split}
 \frac{\partial\hat{\Pi}}{\partial \mu_{B}(T)}=\frac{\partial\hat{\Pi}}{\partial\sigma_{B}(T)}=\frac{\partial\hat{\Pi}}{\partial\eta_{B}(T)}=0.
\end{split}
\end{equation}

Case 4: When $c=-1,K<-\tau$, we have
\begin{equation}
\begin{split}
 \hat{\Pi}&\approx-e^{-rT+m+\frac{1}{2}s^{2}}\Phi(d_{21})+e^{-rT}(-K-\tau)\Phi(d_{22})\\
 &=-e^{-rT}\frac{\sigma_{B}(T)}{\sqrt{x(x-1)}}\sqrt{x}\Phi(d_{21})+e^{-rT}(-K+\mu_{B}(T)+\frac{\sigma_{B}(T)}{\sqrt{x-1}})\Phi(d_{22})\\
 &=-e^{-rT}\frac{\sigma_{B}(T)}{\sqrt{x-1}}\Phi(d_{21})+e^{-rT}(-K+\mu_{B}(T)+\frac{\sigma_{B}(T)}{\sqrt{x-1}})\Phi(d_{22})
\end{split}
\end{equation}
where
\begin{equation}
\begin{split}
 d_{21}&=\frac{\ln(-K-\tau)-m-s^{2}}{s}\\
 &=\frac{\ln(-K+\mu_{B}(T)+\frac{\sigma_{B}(T)}{\sqrt{x-1}})-\ln(\sigma_{B}(T))+\frac{1}{2}\ln(x(x-1))}{(\ln(x))^{\frac{1}{2}}}-(\ln(x))^{\frac{1}{2}},\\
 d_{22}&=\frac{\ln(-K-\tau)-m}{s}\\
 &=\frac{\ln(-K+\mu_{B}(T)+\frac{\sigma_{B}(T)}{\sqrt{x-1}})-\ln(\sigma_{B}(T))+\frac{1}{2}\ln(x(x-1))}{(\ln(x))^{\frac{1}{2}}}.
\end{split}
\end{equation}
The calculations are parallel with the Case 2 above.

\vspace{0.1in}

\textbf{Proof of Lemma \ref{ppp1}:}

The proof is similar to the proof of Lemma \ref{pp2}. We divide into cases.

Case 1: If $\eta_{B}(T)>0$, then $c=1$ and we have
\begin{equation}
\begin{split}
 &M_{1}(T)=e^{m}\phi_{Y_{T}}(\frac{1}{2}s^{2})+\tau,\\
 &M_{2}(T)=e^{2m}\phi_{Y_{T}}(2s^{2})+2\tau e^{m}\phi_{Y_{T}}(\frac{1}{2}s^{2})+\tau^{2},\\
 &M_{3}(T)=e^{3m}\phi_{Y_{T}}(\frac{9}{2}s^{2})+3\tau e^{2m}\phi_{Y_{T}}(2s^{2})+3\tau^{2}e^{m}\phi_{Y_{T}}(\frac{1}{2}s^{2})+\tau^{3}.
\end{split}
\end{equation}
From the above equations, we  get
\begin{equation}
\begin{split}
 M_{2}(T)-M_{1}(T)^{2}=e^{2m}(\phi_{Y_{T}}(2s^{2})-(\phi_{Y_{T}}(\frac{1}{2}s^{2}))^{2})
\end{split}
\end{equation}
and
\begin{equation}
\begin{split}
 M_{3}(T)-&M_{1}(T)^{3}=3\tau e^{2m}(\phi_{Y_{T}}(2s^{2})-(\phi_{Y_{T}}(\frac{1}{2}s^{2}))^{2})+e^{3m}(\phi_{Y_{T}}(\frac{9}{2}s^{2})-(\phi_{Y_{T}}(\frac{1}{2}s^{2}))^{3})\\
 &=3(M_{1}(T)-e^{m}\phi_{Y_{T}}(\frac{1}{2}s^{2}))(M_{2}(T)-M_{1}(T)^{2})+e^{3m}(\phi_{Y_{T}}(\frac{9}{2}s^{2})-(\phi_{Y_{T}}(\frac{1}{2}s^{2}))^{3}),
\end{split}
\end{equation}
which gives
\begin{equation}\label{ppp1_1}
\begin{split}
 M_{3}(T)-&3M_{1}(T)M_{2}(T)+2M_{1}(T)^{3}=\\
 &e^{3m}(\phi_{Y_{T}}(\frac{9}{2}s^{2})-(\phi_{Y_{T}}(\frac{1}{2}s^{2}))^{3})-3e^{m}\phi_{Y_{T}}(\frac{1}{2}s^{2})(M_{2}(T)-M_{1}(T)^{2}).
\end{split}
\end{equation}
If we plug
\begin{equation}
\begin{split}
 e^{m}=\sqrt{\frac{M_{2}(T)-M_{1}(T)^{2}}{\phi_{Y_{T}}(2s^{2})-(\phi_{Y_{T}}(\frac{1}{2}s^{2}))^{2}}}
\end{split}
\end{equation}
into (\ref{ppp1_1}), we obtain
\begin{equation}
\begin{split}
 \frac{\phi_{Y_{T}}(\frac{9}{2}s^{2})-3\phi_{Y_{T}}(\frac{1}{2}s^{2})\phi_{Y_{T}}(2s^{2})+2(\phi_{Y_{T}}(\frac{1}{2}s^{2}))^{3}}{(\phi_{Y_{T}}(2s^{2})-(\phi_{Y_{T}}(\frac{1}{2}s^{2}))^{2})^{\frac{3}{2}}}
 &=\frac{M_{3}(T)-3M_{1}(T)M_{2}(T)+2M_{1}(T)^{3}}{(M_{2}(T)-M_{1}(T)^{2})^{\frac{3}{2}}}\\
 &=\eta_{B}(T),
\end{split}
\end{equation}
which leads to
\begin{equation}
\begin{split}
 \phi_{Y_{T}}(\frac{9}{2}s^{2})-3\phi_{Y_{T}}(\frac{1}{2}s^{2})\phi_{Y_{T}}(2s^{2})+2(\phi_{Y_{T}}(\frac{1}{2}s^{2}))^{3}-\eta_{B}(T)(\phi_{Y_{T}}(2s^{2})-(\phi_{Y_{T}}(\frac{1}{2}s^{2}))^{2})^{\frac{3}{2}}=0.
\end{split}
\end{equation}
If we let $x=s^{2}$, then
\begin{equation}
\begin{split}
 m=\frac{1}{2}\ln(\frac{M_{2}(T)-M_{1}(T)^{2}}{\phi_{Y_{T}}(2x)-(\phi_{Y_{T}}(\frac{1}{2}x))^{2}})=\frac{1}{2}\ln(\frac{\sigma^{2}_{B}(T)}{\phi_{Y_{T}}(2x)-(\phi_{Y_{T}}(\frac{1}{2}x))^{2}}),
\end{split}
\end{equation}
and
\begin{equation}
\begin{split}
 \tau&=M_{1}(T)-e^{m}\phi_{Y_{T}}(\frac{1}{2}s^{2})\\
 &=M_{1}(T)-\sqrt{\frac{M_{2}(T)-M_{1}(T)^{2}}{\phi_{Y_{T}}(2s^{2})-(\phi_{Y_{T}}(\frac{1}{2}s^{2}))^{2}}}\phi_{Y_{T}}(\frac{1}{2}s^{2})\\
 &=\mu_{B}(T)-\frac{\phi_{Y_{T}}(\frac{1}{2}x)}{(\phi_{Y_{T}}(2x)-(\phi_{Y_{T}}(\frac{1}{2}x))^{2})^{\frac{1}{2}}}\sigma_{B}(T).
\end{split}
\end{equation}

Case 2: The case $\eta_{B}(T)<0$ can be handled similar to the above case.

\vspace{0.1in}

\textbf{Proof of Theorem \ref{p3}} We divide into cases.

Case 1:  If $K\leq\tau$, we have
\begin{equation}
\begin{split}
 \hat{\Pi}= e^{-rT}E\left[e^{s\sqrt{Y_{T}}N+m}+\tau-K\right]=e^{m}E[e^{s\sqrt{Y_{T}}N}]+\tau-K.
\end{split}
\end{equation}
We have
\begin{equation}
\begin{split}
 E[e^{s\sqrt{Y_{T}}N}]=E\left[E[e^{s\sqrt{Y_{T}}N}|Y_{T}]\right]=E[e^{\frac{1}{2}s^{2}Y_{T}}]=\phi_{Y_{T}}(\frac{1}{2}s^{2})
\end{split}
\end{equation}
where $\phi_{Y_{T}}(\cdot)$ is the moment generating function of $Y_{T}$. Therefore
\begin{equation}
\begin{split}
 \hat{\Pi}= e^{m}\phi_{Y_{T}}(\frac{1}{2}s^{2})+\tau-K.
\end{split}
\end{equation}

Case 2: If $K>\tau$, the formula becomes
\begin{equation}\label{4}
\begin{split}
 \hat{\Pi}= e^{-rT}E\left[\left(e^{s\sqrt{Y_{T}}N+m}+\tau-K\right)^{+}\right].
\end{split}
\end{equation}
We have
\begin{equation}\label{5}
\begin{split}
 &E\left[\left(e^{s\sqrt{Y_{T}}N+m}+\tau-K\right)^{+}\right]\\
 &=E\left[\left(e^{s\sqrt{Y_{T}}N+m}+\tau-K\right)\cdot I\left(e^{s\sqrt{Y_{T}}N+m}+\tau-K>0\right)\right]\\
 &=E\left[e^{s\sqrt{Y_{T}}N+m}\cdot I\left(e^{s\sqrt{Y_{T}}N+m}+\tau-K>0\right)\right]-(K-\tau)P\left(e^{s\sqrt{Y_{T}}N+m}+\tau-K>0\right),
\end{split}
\end{equation}
where $I(\cdot)$ is the indicator function. The first term of the last expression of (\ref{5}) can be calculated as follows
\begin{equation}
\begin{split}
 &E\left[e^{s\sqrt{Y_{T}}N+m}\cdot I\left(e^{s\sqrt{Y_{T}}N+m}+\tau-K>0\right)\right]\\
 =&E\left[E\left[e^{s\sqrt{Y_{T}}N+m}\cdot I\left(e^{s\sqrt{Y_{T}}N+m}+\tau-K>0\right)|Y_{T}\right]\right]\\
 =&E\left[e^{\frac{1}{2}s^{2}Y_{T}+m}E\left[e^{s\sqrt{Y_{T}}N-\frac{1}{2}s^{2}Y_{T}}\cdot I\left(e^{s\sqrt{Y_{T}}N+m}+\tau-K>0\right)|Y_{T}\right]\right]\\
 =&E\left[e^{\frac{1}{2}s^{2}Y_{T}+m}E\left[I\left(e^{s\sqrt{Y_{T}}(N+s\sqrt{Y_{T}})+m}+\tau-K>0\right)|Y_{T}\right]\right]\\
 =&E\left[e^{\frac{1}{2}s^{2}Y_{T}+m}E\left[I\left(N>\frac{\ln(K-\tau)-m-s^{2}Y_{T}}{s\sqrt{Y_{T}}}\right)|Y_{T}\right]\right]\\
 =&E\left[e^{\frac{1}{2}s^{2}Y_{T}+m}\Phi\left(\frac{-\ln(K-\tau)+m+s^{2}Y_{T}}{s\sqrt{Y_{T}}}\right)\right]. 
\end{split}
\end{equation}
The second term in the last expression of (\ref{5}) can be calculated as follows
\begin{equation}
\begin{split}
 &P\left(e^{s\sqrt{Y_{T}}N+m}+\tau-K>0\right)\\
 =&E\left[P\left(e^{s\sqrt{Y_{T}}N+m}+\tau-K>0|Y_{T}\right)\right]\\
 =&E\left[P\left(N>\frac{\ln(K-\tau)-m}{s\sqrt{Y_{T}}}|Y_{T}\right)\right]\\
 =&E\left[\Phi\left(\frac{-\ln(K-\tau)+m}{s\sqrt{Y_{T}}}\right)\right].
\end{split}
\end{equation}
We let
\begin{equation}\label{6}
\begin{split}
 d_{11}(Y_{T})&=\frac{-\ln(K-\tau)+m+s^{2}Y_{T}}{s\sqrt{Y_{T}}},\\
 d_{12}(Y_{T})&=\frac{-\ln(K-\tau)+m}{s\sqrt{Y_{T}}},
\end{split}
\end{equation}
and plug (\ref{6}) into (\ref{5}) and obtain the result.

Case 3:  If $K\geq-\tau$, since $e^{s\sqrt{Y_{T}}+m}$  and $-\tau-K$ are both non-positive, we have
\begin{equation}
\bar{\Pi}=0. 
\end{equation}

Case 4: If $K<-\tau$, we have
\begin{equation}\label{777}
\begin{split}
 \bar{\Pi}&= e^{-rT}E\left[\left(-e^{s\sqrt{Y_{T}}
 N+m}-\tau-K\right)^{+}\right]\\
 &=-e^{-rT}E\left[\left(e^{s\sqrt{Y_{T}}N+m}+\tau+K\right)^{-}\right]\\
 &=-e^{-rT}E\left[\left(e^{s\sqrt{Y_{T}}N+m}+\tau+K\right)\cdot I\left(e^{s\sqrt{Y_{T}}N+m}+\tau+K<0\right)\right]\\
 &=-e^{-rT}E\left [e^{s\sqrt{Y_{T}}N+m}\cdot I\left(e^{s\sqrt{Y_{T}}N+m}+\tau+K<0\right)\right ]\\
&-(K+\tau)e^{-rT}P\left(
e^{s\sqrt{Y_{T}}N+m}+\tau+K<0
\right ),
\end{split}
\end{equation}
where $I(\cdot)$ is the indicator function. The expectation in the last expression can be evaluated as follows
\begin{equation}
\begin{split}
 &E\left[e^{s\sqrt{Y_{T}}N+m}\cdot I\left(e^{s\sqrt{Y_{T}}N+m}+\tau+K<0\right)\right]\\
 =&E\left[E\left[e^{s\sqrt{Y_{T}}N+m}\cdot I\left(e^{s\sqrt{Y_{T}}N+m}+\tau+K<0\right)|Y_{T}\right]\right]\\
 =&E\left[e^{\frac{1}{2}s^{2}Y_{T}+m}E\left[e^{s\sqrt{Y_{T}}N-\frac{1}{2}s^{2}Y_{T}}\cdot I\left(e^{s\sqrt{Y_{T}}N+m}+\tau+K<0\right)|Y_{T}\right]\right]\\
 =&E\left[e^{\frac{1}{2}s^{2}Y_{T}+m}E\left[I\left(e^{s\sqrt{Y_{T}}(N+s\sqrt{Y_{T}})+m}+\tau+K<0\right)|Y_{T}\right]\right]\\
 =&E\left[e^{\frac{1}{2}s^{2}Y_{T}+m}E\left[I\left(N<\frac{\ln(-K-\tau)-m-s^{2}Y_{T}}{s\sqrt{Y_{T}}}\right)|Y_{T}\right]\right]\\
 =&E\left[e^{\frac{1}{2}s^{2}Y_{T}+m}\Phi\left(\frac{\ln(-K-\tau)-m-s^{2}Y_{T}}{s\sqrt{Y_{T}}}\right)\right].
\end{split}
\end{equation}
The same way, the probability in the second term of the last expression of (\ref{777}) can be simplified as follows
\begin{equation}
\begin{split}
 &P\left(e^{s\sqrt{Y_{T}}N+m}+\tau+K<0\right)\\
 =&E\left[P\left(e^{s\sqrt{Y_{T}}N+m}+\tau+K<0|Y_{T}\right)\right]\\
 =&E\left[P\left(N<\frac{\ln(-K-\tau)-m}{s\sqrt{Y_{T}}}|Y_{T}\right)\right]\\
 =&E\left[\Phi\left(\frac{\ln(-K-\tau)-m}{s\sqrt{Y_{T}}}\right)\right].
\end{split}
\end{equation}
We let
\begin{equation}\label{8}
\begin{split}
 d_{21}(Y_{T})&=\frac{\ln(-K-\tau)-m-s^{2}Y_{T}}{s\sqrt{Y_{T}}},\\
 d_{22}(Y_{T})&=\frac{\ln(-K-\tau)-m}{s\sqrt{Y_{T}}},
\end{split}
\end{equation}
and plug (\ref{8})  into (\ref{777})
and obtain (\ref{p4_3}).

\bibliographystyle{siam}
\bibliography{reference}

\begin{thebibliography}{10}

\bibitem{Barndorff}
{\sc O.~Barndorff-Nielsen}, {\em Normal inverse gaussian distributions and
  stochastic volatility modelling}, Scandinavian Journal of Statistics, 24
  (1997), pp.~1--12.

\bibitem{barndorff-nielsen}
\leavevmode\vrule height 2pt depth -1.6pt width 23pt, {\em Processes of normal
  inverse gaussian type}, Finance and Stochastics, 2 (1998), pp.~41--68.

\bibitem{Bingham}
{\sc N.~Bingham and R.~Kiesel}, {\em Semi-parametric modelling in finance:
  theoretical foundations}, Quantitative Finance, 2 (2002), pp.~241--250.

\bibitem{Bjerk_2014}
{\sc P.~Bjerksund and G.~Stensland}, {\em Closed form spread option valuation},
  Quantitative Finance, 14 (2014), pp.~1785--1794.

\bibitem{Borovkova}
{\sc S.~Borovkova, F.~Permana, and H.~Weide}, {\em A closed form approach to
  the valuation and hedging of basket and spread option}, The Journal of
  Derivatives, 14 (2007), pp.~8--24.

\bibitem{Carmona2003b}
{\sc R.~Carmona and V.~Durrleman}, {\em Pricing and hedging spread options},
  Siam Review, 45 (2003), pp.~627--685.

\bibitem{Carmona2003a}
\leavevmode\vrule height 2pt depth -1.6pt width 23pt, {\em Pricing and hedging
  spread options in a log-normal model (technical report: Department of
  operations research and financial engineering)}, Princeton, NJ: Princeton
  University,  (2003).

\bibitem{cont-tankov}
{\sc R.~Cont and P.~Tankov}, {\em Financial modelling with jump processes},
  Chapman and hall-CRC financial mathematics series, 2004.

\bibitem{Eberlein}
{\sc E.~Eberlein and U.~Keller}, {\em Hyperbolic distributions in finance},
  Bernoulli, 1 (1995), pp.~281--299.

\bibitem{JARROW_RUDD}
{\sc R.~Jarrow and A.~Rudd}, {\em Approximate option valuation for arbitrary
  stochastic processes}, Journal of financial Economics, 10 (1982),
  pp.~347--369.

\bibitem{Emmanuel}
{\sc E.~Jurczenko, B.~Maillet, and B.~Negrea}, {\em A note on skewness and
  kurtosis adjusted option pricing models under the martingale restriction},
  Quantitative Finance, 4 (2004), pp.~479--488.

\bibitem{Kirk1995}
{\sc E.~Kirk and J.~Aron}, {\em Correlation in the energy markets}, Managing
  energy price risk, 1 (1995), pp.~71--78.

\bibitem{kozubowski-podgorski}
{\sc T.~J. Kozubowski and K.~Podgorski}, {\em Asymmetric laplace laws and
  modeling financial data}, Mathematical and Computer Modelling, 34 (2001),
  pp.~1003--1021.

\bibitem{Nicolas_Geoffrey}
{\sc N.~Langrene, G.~Lee, and Z.~Zhu}, {\em Switching to non-affine stochastic
  volatility: A closed-form expansion for the inverse gamma model},
  International Journal of Theoretical and Applied Finance, 19 (2016).

\bibitem{Leccadito}
{\sc A.~Leccadito, T.~Paletta, and R.~Tunaru}, {\em Pricing and hedging basket
  options with exact moment matching}, Insurance: Mathematics and Economics, 69
  (2016), pp.~59--69.

\bibitem{Levy_1992}
{\sc E.~L\'evy}, {\em The valuation of average rate currency options}, Journal
  of International Money and Finance, 11 (1992), pp.~474--491.

\bibitem{Li_Deng_2008}
{\sc M.~Li, S.~Deng, and J.~Zhoc}, {\em Closed-form approximations for spread
  option prices and greeks}, The Journal of Derivatives, 15 (2008), pp.~58--80.

\bibitem{Li_Deng_2009}
{\sc M.~Li, J.~Zhoc, and S.~Deng}, {\em Multi-asset spread option pricing and
  hedging}, Quantitative Finance, 10 (2009), pp.~305--324.

\bibitem{luciano-schoutens}
{\sc E.~Luciano and W.~Schoutens}, {\em A multivariate jump-driven financial
  asset model}, Quantitative Finance, 6 (2006), pp.~385--402.

\bibitem{madan-sieneta}
{\sc D.~B. Madan and E.~Seneta}, {\em The variance gamma model for share market
  returns}, Journal of Business, 63 (1990), pp.~511--524.

\bibitem{Prause}
{\sc K.~Prause et~al.}, {\em The generalized hyperbolic model: Estimation,
  financial derivatives, and risk measures}, PhD thesis,
  Albert-Ludwigs-Universität Freiburg, 1999.

\bibitem{Raible}
{\sc S.~Raible}, {\em L{\'e}vy processes in finance: Theory, numerics, and
  empirical facts}, PhD thesis, Albert-Ludwigs-Universität Freiburg, 2000.

\bibitem{Schoutens}
{\sc W.~Schoutens}, {\em Le\'vy Processes in Finance: Pricing Financial
  Derivatives}, John Wiley $\&$ Sons, Ltd, 2003.

\bibitem{Shimko_1994}
{\sc D.~Shimko}, {\em Options on futures spreads-hedging, speculation and
  valuation}, Journal of Futures Markets, 14 (1994), pp.~183--213.

\bibitem{Turnbull1991}
{\sc S.~Turnbull and L.~Wakeman}, {\em A quick algorithm for pricing european
  average options}, The Journal of Financial and Quantitative Analysis, 26
  (1991), pp.~377--389.

\bibitem{Wu_Diao2019}
{\sc F.~Wu, X.~Diao, and C.~Wu}, {\em A general accurate approximation for
  pricing and hedging basket options with exact moment matching}, The Journal
  of Derivatives, 26 (2019), pp.~68--86.

\end{thebibliography}

\end{document}